\documentclass[12pt]{amsart}

\usepackage{amssymb, enumerate, mdwlist}

\evensidemargin\paperwidth
\advance\evensidemargin-\textwidth
\advance\oddsidemargin-1in
\evensidemargin\oddsidemargin

\topmargin\paperheight
\advance\topmargin-\textheight
\advance\topmargin-1in

\theoremstyle{plain}
\newtheorem{theorem}{Theorem}[section]
\theoremstyle{plain}
\newtheorem{proposition}[theorem]{Proposition}
\theoremstyle{plain}
\newtheorem{lemma}[theorem]{Lemma}
\theoremstyle{plain}
\newtheorem{corollary}[theorem]{Corollary}
\theoremstyle{plain}

\theoremstyle{plain}
\newtheorem{definition}[theorem]{Definition}
\theoremstyle{plain}

\theoremstyle{remark}
\newtheorem{remark}[theorem]{Remark}
\theoremstyle{remark}
\newtheorem{example}[theorem]{Example}
\theoremstyle{remark}

\title
[Analytic homogeneous QW on $\mathbb{Z}$]
{Space-homogeneous quantum walks on $\mathbb{Z}$\\
from the viewpoint of complex analysis}
\author{Hayato Saigo}
\address
{Nagahama Institute of Bio-Sciences and Technology, Tamura, Nagahama, 526-0829, Japan}
\email
{h\_saigoh@nagahama-i-bio.ac.jp}
\author{Hiroki Sako}
\address
{Faculty of Engineering, Niigata University, Nishi-ku, Niigata 950-2181, Japan}
\email
{sako@eng.niigata-u.ac.jp}
\subjclass[2010]{46L99, 60F05, 81Q99}

\begin{document}
\maketitle

\begin{abstract}
The subject of this paper is quantum walks, which are expected to simulate several kinds of quantum dynamical systems.
In this paper, we define {\it analyticity} for quantum walks on $\mathbb{Z}$. 
Almost all the quantum walks on $\mathbb{Z}$ which have been already studied are analytic.
In the framework of analytic quantum walks,
we can enlarge the theory of quantum walks.
We
obtain not only several generalizations of known results, but also new types of theorems.
It is proved that every analytic space-homogeneous quantum walk on $\mathbb{Z}$ is essentially a composite of shift operators and continuous-time analytic space-homogeneous quantum walks. 
We also prove existence of the weak limit distribution for analytic space-homogeneous quantum walks on $\mathbb{Z}$.
\end{abstract}

\section{Introduction}

In this paper, we study a kind of dynamical systems called {\it quantum walks}.
Many researchers have already studied the subject in several different frameworks (see \cite{Meyer}, \cite{ABNVW} for example). 
They commonly make use of the following items:
\begin{itemize}
\item
the Hilbert space  $\mathcal{H} = \ell_2 (X) \otimes \mathbb{C}^n$
defined on a (discrete) metric space $X$.
\item
a unitary operator $U$ on $\ell_2 (X) \otimes \mathbb{C}^n$, 
\item
and a unit vector $\xi$ in the Hilbert space.
\end{itemize}
The sequence (or $1$-parameter family) of unit vectors $\{U^t \xi\}_{t}$ defines a probability measure on the space $X$, which has attracted much attention.

In this paper, we focus on the case that the space $X$ is given by the set $\mathbb{Z}$ of integers and that the unitary operator $U$ is space-homogeneous.
The following are aims of this paper:
\begin{itemize}
\item
We will prove a structure theorem on such a walk $U$ (Theorem \ref{theorem: structure theorem}).
\item
The walk $U$ can be constructed from shift operators and continuous-time space-homogeneous quantum walks on $\mathbb{Z}$ (Theorem \ref{theorem: construction from CTQW}).
\item
We will prove that the eigenvalue functions introduced in Definition \ref{definition: system of eigenvalue functions} determines when $U$ is a restriction of a continuous-time space-homogeneous quantum walk (Theorem \ref{theorem: realization by CTQW}).
\item
The walk $U$ always has weak limit distribution for every initial unit vector $\xi$ which rapidly decreases. Its precise statement is given in Theorem \ref{theorem: limit distribution}.
\end{itemize}
Theorem \ref{theorem: structure theorem} means that every space-homogeneous quantum walk on $\mathbb{Z}$ is essentially a direct sum of {\it model quantum walks}.
The model quantum walk $U_{d, \lambda}$ is
introduced in subsection \ref{subsection: model quantum walk}.
They are labeled by a pair of a natural number $d$ and an analytic map $\lambda \colon \mathbb{T} \to \mathbb{T}$.
Our new framework of quantum walks is large enough to include all the model quantum walks.

To state above theorems in this paper, we need to clarify the definition of quantum walks.
Throughout this paper, we always require {\it analyticity} in the sense of Definition \ref{definition: analyticity} for quantum walks. 
This assumption is so weak that almost all known examples satisfy.
We need complex analysis on the Fourier dual of the quantum walk.
Requirement of analyticity on the walk enables us to use Riemann surfaces.
Many mathematicians and physicists have already used Fourier analysis on quantum walks.
Combining with complex analysis, we can extend the study further.

In Section \ref{section: convergence theorem}, we prove existence of the weak limit distribution for every analytic space-homogeneous quantum walk.
Grimmet, Janson, and Scudo stated this theorem in \cite{GJS} and 
our argument follows their excellent idea.
The paper \cite{GJS} does not explicitly define quantum walks.
This is one of the reasons why it is difficult for the readers to check 
the claims in \cite{GJS}.
They put explicit and implicit assumptions (see Remark
\ref{remark: our advantage}).
In Subsection \ref{subsection: modified Hadamard walk},
we construct an example 
for which an implicit assumption in \cite{GJS} does not hold.
The authors think that it is unnatural to exclude the example
from the class of quantum walks.

In Section \ref{section: algebraic definition}, 
we prove that every analytic space-homogeneous quantum walk on $\mathbb{Z}$ is a solution of algebraic equation, whose coefficients are elements of an operator algebra.
The authors expect that there might be more algebraic way of the definition of quantum walks, which enlarges the scope of our study further.

In Section \ref{section: examples}, we examine a new example of quantum walks, as well as known examples. If the readers want to start with concrete examples, the authors recommend them to see Section \ref{section: examples} first.

\section{Preliminary on vector-valued analytic maps}
This paper completely relies on complex analysis.
To study the inverse Fourier dual of quantum walk,
in Subsection \ref{subsection: analytic section}, we construct analytic sections of eigenvectors.
For the argument, we prepare a couple of lemmata.
Let $\mathbb{T}$ be the set of complex numbers whose absolute values are $1$.

\begin{lemma}\label{lemma: unit vector}
Let ${\bf x} \colon \mathbb{T} \to \mathbb{C}^n$ be an analytic map.
Suppose that the map $\bf x$ is not the constant map $\bf 0$.
Then there exists an analytic map
${\bf v} \colon  \mathbb{T} \to \mathbb{C}^n$ satisfying the following conditions:
\begin{itemize}
\item
for every $z \in \mathbb{T}$, $\|{\bf v}(z)\| = 1$,
\item
for every $z \in \mathbb{T}$, ${\bf x}(z) \in \mathbb{C} {\bf v}(z)$.
\end{itemize}
\end{lemma}
We call $\bf v$ {\it a normalization} of $\bf x$.

\begin{proof}
For $k = 1, 2, \cdots, n$, let $x_k \colon \mathbb{T} \to \mathbb{C}$ be the $k$-th entry of
the analytic map ${\bf x} \colon \mathbb{T} \to \mathbb{C}^n$.
There exists an open neighborhood $\Omega$ of $\mathbb{T}$ such that
$\Omega$ is invariant under the reflexion $z \mapsto \dfrac{1}{\ \overline z \ }$ and
that $x_k \colon \mathbb{T} \to \mathbb{C}$ admits an analytic (or holomorphic) extension
\[x_k \colon \Omega \to \mathbb{C},\]
for every $k$.
Define $x_k^* \colon \Omega \to \mathbb{C}$ by the reflection $\overline {x_k \left( \dfrac{1}{\ \overline z\ } \right)}$. Note 
that $x_k^*$ is analytic on $\Omega$ and
that the equation $x_k^* = \overline{x_k}$ holds on $\mathbb{T}$.

Consider the analytic function
\[z \mapsto \sum_{k = 1}^n x_k^*(z) x_k(z) \]
defined on $\Omega$.
If the function has no zero on $\mathbb{T}$, then the map
\[{\bf v}(z) = \dfrac{{\bf x}(z)}{\sqrt{\sum_{k = 1}^n x_k^*(z) x_k(z)}},\]
satisfies the conditions in the lemma.

Consider the case that 
there exists a zero of analytic function
\[z \mapsto \sum_{k = 1}^n x_k^*(z) x_k(z) \]
on $\mathbb{T}$.
Note that the order of the zero on $\mathbb{T}$ is even. There exist analytic functions
\begin{eqnarray*}
\| {\bf x} \|_+ \colon \mathbb{T}\setminus \{-1\} &\to& \mathbb{R},\\
\| {\bf x} \|_- \colon \mathbb{T}\setminus \{1\} &\to& \mathbb{R}.
\end{eqnarray*}
satisfying that
\begin{eqnarray*}
\| {\bf x} \|_+(z)^2 &=& \sum_{k = 1}^n x_k^*(z) x_k(z), 
\quad z \in \mathbb{T}\setminus \{-1\},\\
\| {\bf x} \|_-(z)^2 &=& \sum_{k = 1}^n x_k^*(z) x_k(z),
\quad z \in \mathbb{T}\setminus \{1\}.
\end{eqnarray*}
and that $\| {\bf x} \|_+(z) = \| {\bf x} \|_-(z)$ on the intersection of the upper half plane and $\mathbb{T}$.
Note that on the intersection of the lower half plane and $\mathbb{T}$, $\| {\bf x} \|_-(z)$ is identical to $\| {\bf x} \|_+(z)$ or $-\| {\bf x} \|_+(z)$.

If $z \in \mathbb{T}$ is a zero of $\| {\bf x} \|_+$ or $\| {\bf x} \|_-$, then the order coincides with
\[\min_{1 \le k \le n} (\textrm{the\ order\ of\ zero\ of\ } x_k \textrm{\ at\ } z).\]
This means that the singular points of
\[z \mapsto \dfrac{{\bf x}(z)}{ \| {\bf x} \|_+(z)}, 
\quad z \mapsto \dfrac{{\bf x}(z)}{ \| {\bf x} \|_-(z)}, \]
on $\mathbb{T}$ are removable.
If $\| {\bf x} \|_-(z)$ is identical to $\| {\bf x} \|_+(z)$ on the intersection of the lower half plane and $\mathbb{T}$, define ${\bf v}(z)$ by 
\[\dfrac{{\bf x}(z)}{ \| {\bf x} \|_+(z)} = \dfrac{{\bf x}(z)}{ \| {\bf x} \|_-(z)}.\]
We finish the proof in such a case.

Consider the case that
\begin{eqnarray*}
\| {\bf x} \|_-(z) &=& \| {\bf x} \|_+(z), \quad z \in \mathbb{T}, \mathrm{Im}(z) > 0,\\
\| {\bf x} \|_-(z) &=& - \| {\bf x} \|_+(z), \quad z \in \mathbb{T}, \mathrm{Im}(z) < 0.
\end{eqnarray*}
Define $\bf v$ by 
\begin{eqnarray*}
{\bf v}(\exp(i\theta_1)) 
&=& \exp \left( \dfrac{i \theta_1}{2} \right) \dfrac{{\bf x}(\exp(i \theta_1))}{ \| {\bf x} \|_+(\exp(i \theta_1))},
\quad - \pi < \theta_1 < \pi,
\\
{\bf v}(\exp(i\theta_2)) 
&=& \exp \left( \dfrac{i \theta_2}{2} \right) \dfrac{{\bf x}(\exp(i \theta_2))}{ \| {\bf x} \|_-(\exp(i \theta_2))},
\quad 0 < \theta_2 < 2\pi.
\end{eqnarray*}
This defines a single-valued function, because for every $-\pi < \theta_1 < \pi$ and $\pi < \theta_2 < 2 \pi$, if $\theta_2 - \theta_1 = 2 \pi$, then $\exp(i \theta_1/2) = - \exp(i \theta_2/2)$.
\end{proof}

\begin{lemma}\label{lemma: orthonormal basis}
Let ${\bf x}^{(1)}, {\bf x}^{(2)}, \cdots, {\bf x}^{(d)} \colon \mathbb{T} \to \mathbb{C}^n$ be a collection of analytic maps.
Suppose that on a coset of a finite subset of $\mathbb{T}$, the vectors
\[\{ {\bf x}^{(1)}(z), {\bf x}^{(2)}(z), \cdots, {\bf x}^{(d)}(z) \}\] 
are linearly independent.
Then there exist analytic maps
\[{\bf v}^{(1)}, {\bf v}^{(2)}, \cdots, {\bf v}^{(d)} \colon \mathbb{T} \to \mathbb{C}^n\] satisfying the following conditions:
\begin{itemize}
\item
for every $z \in \mathbb{T}$, $\{ {\bf v}^{(1)}(z), {\bf v}^{(2)}(z), \cdots, {\bf v}^{(d)}(z) \}$ forms an orthonormal system,
\item
for every $z \in \mathbb{T}$, ${\bf x}^{(1)}(z), {\bf x}^{(2)}(z), \cdots, {\bf x}^{(d)}(z)$ are elements of  the linear span of $\{{\bf v}^{(1)}(z), {\bf v}^{(2)}(z), \cdots, {\bf v}^{(d)}(z)\}$.
\end{itemize}
\end{lemma}

\begin{proof}
The Gram--Schmidt process works in our framework.

Because the map ${\bf x}^{(1)}$ is analytic and not the constant map $\bf 0$,
by Lemma \ref{lemma: unit vector},
there exists an analytic map ${\bf v}^{(1)}$ which is a normalization of ${\bf x}^{(1)}$.
Note that on a coset of a finite subset of $\mathbb{T}$, $\mathbb{C} {\bf x}^{(1)} (z) = \mathbb{C} {\bf v}^{(1)}(z)$.

On $\mathbb{T}$, define ${\bf y}^{(2)}(z)$ by 
\[{\bf y}^{(2)}(z) = {\bf x}^{(2)}(z) 
- \langle{\bf x}^{(2)} (z),  {\bf v}^{(1)}(z) \rangle {\bf v}^{(1)}(z). \]
For $k = 1, 2, \cdots, n$, the complex conjugate $\overline{v^{(1)}_k(z)}$ 
of the $k$-th entry of ${\bf v}^{(1)}(z)$ is an analytic function on $\mathbb{T}$, 
because it is identical to 
$\overline{v^{(1)}_k \left( \dfrac{1}{\ \overline z \ } \right)}$.
It follows that the map ${\bf y}^{(2)}(z)$ is also an analytic map on $\mathbb{T}$.
Because on a coset of finite subset of $\mathbb{T}$,
${\bf x}^{(2)}(z)$ is linearly independent of $\mathbb{C} {\bf v}^{(1)}(z)$,
${\bf y}^{(2)}$ is not the constant map $\bf 0$.
Again by Lemma \ref{lemma: unit vector}, 
there exists an analytic map ${\bf v}^{(2)}$ which is a normalization of ${\bf y}^{(2)}$.
Because ${\bf v}^{(2)}(z)$ is perpendicular to ${\bf v}^{(1)}(z)$
on a coset of a finite subset of $\mathbb{T}$, $\{{\bf v}^{(1)}(z), {\bf v}^{(2)}(z)\}$ forms an orthonormal system on the coset.
By the continuity, the vectors form an orthonormal system for every  $z \in \mathbb{T}$.
On $\mathbb{T}$, we have
\[{\bf x}^{(1)}(z) \in \mathbb{C}{\bf v}^{(1)}(z), \quad {\bf y}^{(2)}(z) = {\bf x}^{(2)}(z) 
- \langle {\bf x}^{(2)} (z), {\bf v}^{(1)}(z) \rangle {\bf v}^{(1)}(z) \in \mathbb{C}{\bf v}^{(2)}(z).\]
It follows that
\[\mathrm{span} \{{\bf x}^{(1)}(z), {\bf x}^{(2)}(z)\} \subset \mathrm{span}\{ {\bf v}^{(1)}(z), {\bf v}^{(2)}(z)\}.\]
On a coset of a finite subset of $\mathbb{T}$, the vectors ${\bf x}^{(1)}(z)$, 
${\bf x}^{(2)}(z)$ are linearly independent, and therefore the above two subspaces coincide.

On $\mathbb{T}$, define ${\bf y}^{(3)}(z)$ by 
\[{\bf y}^{(3)}(z) = {\bf x}^{(3)}(z) 
- \langle {\bf x}^{(3)} (z), {\bf v}^{(1)}(z) \rangle {\bf v}^{(1)}(z)
- \langle {\bf x}^{(3)} (z), {\bf v}^{(2)}(z) \rangle {\bf v}^{(2)}(z),
\quad z \in \mathbb{T}. \]
Since the complex conjugates of the entries of ${\bf v}^{(1)}$, ${\bf v}^{(2)}$ are analytic functions on $\mathbb{T}$, 
the map ${\bf y}^{(3)}$ is also analytic on $\mathbb{T}$.
Because on a coset of a finite subset of $\mathbb{T}$,
${\bf x}^{(3)}(z)$ is linearly independent of 
$ \mathrm{span}\{ {\bf v}^{(1)}(z), {\bf v}^{(2)}(z)\}
=
\mathrm{span} \{{\bf x}^{(1)}(z), {\bf x}^{(2)}(z)\}$,
${\bf y}^{(3)}$ is not the constant map $\bf 0$.
Again by Lemma \ref{lemma: unit vector},
there exists an analytic map ${\bf v}^{(3)}$ which is a normalization of ${\bf y}^{(3)}$.
On a coset of a finite subset of $\mathbb{T}$, $\{{\bf v}^{(1)}(z), {\bf v}^{(2)}(z), {\bf v}^{(3)}(z)\}$ forms an orthonormal system.
By the continuity, on $\mathbb{T}$, the system is orthonormal.
On $\mathbb{T}$, we have
${\bf y}^{(3)}(z) 
\in \mathbb{C}{\bf v}^{(3)}(z)$.
It follows that
\[\mathrm{span} \{{\bf x}^{(1)}(z), {\bf x}^{(2)}(z), {\bf x}^{(3)}(z)\} 
\subset 
\mathrm{span}\{ {\bf v}^{(1)}(z), {\bf v}^{(2)}(z), {\bf v}^{(3)}(z)\}.\]
On a coset of a finite subset of $\mathbb{T}$, the above two subspaces coincide.

Repeating this procedure, we obtain an orthonormal system in the lemma.
\end{proof}

We note that two subspaces
\[\mathrm{span} \{{\bf x}^{(1)}(z), \cdots, {\bf x}^{(d)}(z)\} \subset \mathrm{span}\{ {\bf v}^{(1)}(z), \cdots, {\bf v}^{(d)}(z)\}\]
coincide on a coset of a finite subset of $\mathbb{T}$.

\section{Quantum walks on $\ell_2(\mathbb{Z}) \otimes \mathbb{C}^n$ and Analyticity}
\label{section: analyticity}

The subject of this paper is a unitary operator satisfying the following conditions.

Let $n$ be a natural number.
Consider a bounded linear operator $X$ on $\ell_2(\mathbb{Z}) \otimes \mathbb{C}^n$.
The matrix expression
\[ [X((s, k), (t, l))]_{(s, k), (t, l) \in \mathbb{Z} \times \{1, 2, \cdots, n\}}\]
of $X$
is given by
\[X((s, k), (t, l)) = \langle X (\delta_t \otimes \delta_l), \delta_s \otimes \delta_k\rangle.\]
\begin{definition}\label{definition: analyticity}
\begin{enumerate}
\item\label{item of definition: C-infty}
The operator $X$ is said to be {\rm in the C$^\infty$-class}, if
for every natural number $N$, the set
\[ \{ (1 + |s - t|^2)^N X((s, k), (t, l)) \ |\ s, t \in \mathbb{Z}, k, l \in \{1, 2, \cdots, n\} \}\]
is bounded.
\item\label{item of definition: analytic}
The operator $X$ is said to be {\rm analytic}, if
there exist constants $0 < c$ and $1 < r$ satisfying that
for every $k, l, s, t$
\[|X((s, k), (t, l))| \le c r^{-|s - t|}\]
\item\label{item of definition: finite propagation}
The operator $X$ is said to have {\rm finite propagation}, if
there exists a constant $1 \le R$ satisfying that
for every $k, l, s, t$
\[X((s, k), (t, l)) = 0,\]
whenever $|s - t|$ is greater than $R$.
(see e.g. \cite[Definition 5.9.2]{NowakYu})
\item
The operator $X$ is said to be {\rm homogeneous} or {\rm space-homogeneous}, if
the matrix coefficient 
$X((s, k), (t, l))$
depends only on $k, l,$ and $s - t$.
\end{enumerate}
\end{definition}

Note that for an operator $X$ on $\ell_2(\mathbb{Z}) \otimes \mathbb{C}^n$.
\begin{itemize}
\item
The operator $X$ is analytic, if it has finite propagation.
\item
The operator $X$ is in the C$^\infty$-class, if it is analytic.
\item
The operator $X$ is bounded, if it is in the C$^\infty$-class.
\end{itemize}

We often identify a vector in $\ell_2(\mathbb{Z}) \otimes \mathbb{C}^n$ with a column vector
whose entries are $\ell_2$ functions on $\mathbb{Z}$.
Every bounded linear operator on $\ell_2(\mathbb{Z}) \otimes \mathbb{C}^n = \ell_{2}(\mathbb{Z})^n$ can be expressed by an $(n \times n)$-matrix whose entries are bounded linear operator acting on $\ell_2(\mathbb{Z})$.

For $s \in \mathbb{Z}$, let $S_s$ be the unitary operator on $\ell_2(\mathbb{Z})$ defined by the shift
\[\delta_t \mapsto \delta_{s + t}, \quad t \in \mathbb{Z}.\] 
A homogeneous operator $X = [X((s, k), (t, l))]$ can be expressed by infinite sums
\[ \left[ \sum_{s} X_{k, l}(s) S_s \right]_{1 \le k, l \le n}, \]
where the coefficient $X_{k, l}(s - t)$ is given by $X((s, k), (t, l))$.
In the case that $X$ is in the C$^\infty$-class, the infinite sums are given by operator norm convergence.

It is not hard to see that the following classes of unitary operators form groups:
\begin{itemize}
\item
unitary operators in the C$^\infty$-class,
\item
analytic unitary operators,
\item
unitary operators with finite propagation,
\item
homogeneous unitary operators.
\end{itemize}

%

We give a definition of quantum walks on $\mathbb{Z}$.

\begin{definition}
Let $r \mapsto U^{(t)}$
be a group homomorphism from a closed subgroup $G \in \mathbb{R}$ to the group consisting of unitary operators acting on $\ell_2(\mathbb{Z}) \otimes \mathbb{C}^n$.
Note that the group $G$ is $\mathbb{R}$ or of the form $c \mathbb{Z}$.
\begin{itemize}
\item
The group homomorphism is called {\rm a continuous-time quantum walk on} $\mathbb{Z}$, if
the group $G$ is $\mathbb{R}$ and the homomorphism is continuous with respect to weak operator topology.
\item
The group homomorphism is called {\rm a discrete-time quantum walk} on $\mathbb{Z}$
or simply {\rm a quantum walk} on  $\mathbb{Z}$, 
if the group $G$ is of the form $c \mathbb{Z}$.
\item
The quantum walk is said to be in  {\rm the C$^\infty$-class},
if $U^{(t)}$ is a unitary operator in the ${\rm C}^\infty$-class for every $t \in G$.
\item
The quantum walk is said to be  {\rm analytic},
if $U^{(t)}$ is an analytic unitary operator for every $t \in G$.
\item
The quantum walk is said to have  {\rm finite propagation},
if $U^{(t)}$ has finite propagation for every $t \in G$.
\item
The quantum walk is said to be  {\rm space-homogeneous}, or more simply {\rm homogeneous}
if $U^{(t)}$ is a homogeneous unitary operator for every $t \in G$.
\end{itemize}
\end{definition}

\begin{example}
Almost all discrete-time space-homogeneous quantum walks on $\mathbb{Z}$ have finite propagation.
Therefore, they are analytic unitary operators.
For example, the Hadamard quantum walk have finite propagation.
See \cite{ABNVW} for examples of quantum walks on $\mathbb{Z}$.
\end{example}

\section{Fourier and complex analysis for quantum walks on $\mathbb{Z}$}

In this section, we define and discuss the inverse Fourier transform $\widehat{U}$ of a discrete-time analytic homogeneous quantum walk $U$ on $\mathbb{Z}$.
Such a transform defines a map $z \mapsto \widehat{U}(z)$ from $\mathbb{T}$ to unitary operators acting on $\mathbb{C}^n$.
The goal of this section is to construct analytic sections of eigenvectors of $\widehat{U}(z)$ which form an orthonormal basis of each fiber $\mathbb{C}^n$ .
The main ingredient is analyticity of the map $z \mapsto \widehat{U}(z)$.

\subsection{Inverse Fourier transform for homogeneous operators and quantum walks}

Our discussion will completely rely on Fourier analysis.
We express the Pontryagin dual of $\mathbb{Z}$ by
$\mathbb{T} = \{z \in \mathbb{C} \ |\ |z| = 1\}$.
The inverse Fourier transform $\mathcal{F}^{-1}$ is given by the unitary operator
\[{\mathcal F}^{-1} \colon \ell_2(\mathbb{Z}) \ni \delta_s \mapsto z^s \in L^2(\mathbb{T}).\]
For a bounded linear operator $X$ on $\ell_2(\mathbb{Z}) \otimes \mathbb{C}^n$,
we call the bounded linear operator $\widehat{X} = (\mathcal{F}^{-1} \otimes \mathrm{id}) X 
(\mathcal{F} \otimes \mathrm{id})$
acting on $L^2(\mathbb{T}) \otimes \mathbb{C}^n$
{\it the inverse Fourier transform} of $X$.

In the case that $X$ is a homogeneous operator, $X$ can be expressed as
\[ X = \left[ \sum_{s \in \mathbb{Z}} X_{k, l}(s) S_s \right]_{1 \le k, l \le n}, \]
and the $(k, l)$-entry of the inverse Fourier transform $\widehat{X}$ is the multiplication operator
by the function 
\[\sum_{s \in \mathbb{Z}} X_{k, l}(s) z^s \in L^\infty(\mathbb{T}) \]
on $\mathbb{T}$.
For a continuous function $\eta$ on $\mathbb{T}$, we often identify the function $\eta$ and the multiplication operator
\[L^2(\mathbb{T}) \ni \xi \mapsto \eta \xi \in L^2(\mathbb{T}).\]
Sometimes, to emphasize that $\eta$ gives a multiplication operator,
we denote by $M[\eta]$ the operator.

\begin{lemma}\label{lemma: Fourier expansions}
For the homogeneous operator $X$ on $\ell_2(Z) \otimes \mathbb{C}^n$, we have the following
\begin{enumerate}
\item
The operator $X$ is in the {\rm C}$^\infty$-class, if and only if every entry of $\widehat{X}$ is a smooth function. 
\item
The operator $X$ is analytic, if and only if every entry of $\widehat{X}$ is an analytic function defined on a neighborhood of $\mathbb{T}$.
\item
The operator $X$ has finite propagation, if and only if every entry of $\widehat{X}$ is 
a linear combination of $\{z^s \ |\ s \in \mathbb{Z}\} \subset C(\mathbb{T})$.
\end{enumerate}
\end{lemma}

\begin{proof}
Smoothness and analyticity of functions on $\mathbb{T}$ can be rephrased by how rapidly  the Fourier coefficients decreases.
\end{proof}

For the rest of this paper, we focus on a discrete-time analytic homogeneous quantum walk acting on $\ell_2(\mathbb{Z}) \otimes \mathbb{C}^n$.
We may assume that the group of time is $\mathbb{Z}$.
Then we simply denote the generator of the quantum walk $U^{(1)}$ by $U$.
The inverse Fourier transform 
$\widehat{U} 
= ({\mathcal F}^{-1} \otimes {\mathrm id}) U ({\mathcal F} \otimes {\mathrm id})$ 
of the generator of the quantum walk is an element of 
$M_n(C(\mathbb{T})) = C(\mathbb{T}) \otimes M_n(\mathbb{C})$.
We denote by $\widehat{U}(z ; k, l)$ the $(k, l)$-entry of $\widehat{U}$. 
By Lemma \ref{lemma: Fourier expansions},
the function $\widehat{U}(z ; k, l)$ has an analytic extension to a neighborhood of $\mathbb{T}$.
We also denote by $\widehat{U}(z ; k, l)$ the extension.
Note that for every $z \in \Omega$, $\left( \widehat{U}(z ; k, l) \right)_{k, l}$ is an $(n \times n)$ unitary matrix and 
that for every $z \in \mathbb{T}$, $\left( \widehat{U}(z ; k, l) \right)_{k, l}$ is a unitary matrix.

\subsection{Polynomials whose coefficients are analytic functions on $\mathbb{T}$}

For the rest of this section,
let $f(\lambda; z)$ be the characteristic polynomial 
\[f(\lambda; z) = \det\left( \left(\lambda \delta_{k, l} - \widehat{U}(z ; k, l) \right)_{k, l} \right).\]
The degree of polynomial with respect to $\lambda$ is $n$. The coefficients are analytic (or holomorphic) functions of $z$ defined on a domain containing $\mathbb{T}$.
To study such a polynomial, we need some preparation.

\begin{definition}
\begin{itemize}
\item
Let $Q_\mathbb{T}$ be the set of all the pairs $(\Omega, q)$ of a domain $\Omega$ containing  $\mathbb{T}$ and an analytic map to the Riemann sphere
\[q \colon \Omega \to \mathbb{C} \cup \{ \infty \} \]
which is not the constant map $\infty$.
\item
Two elements $(\Omega_1, q_1)$, $(\Omega_2, q_2)$ of $Q_\mathbb{T}$ are said to be equivalent, if 
there exists a domain $\mathbb{T} \subset \Omega_0 \subset \Omega_1 \cap \Omega_2$ on which $q_1$ and $q_2$ coincide. 
\item
Let ${\mathcal Q}_\mathbb{T}$ be the set of all the equivalence classes with respect to the above equivalence relation. 
\end{itemize}
\end{definition}

\begin{lemma}
The point-wise summation and multiplication of $Q_\mathbb{T}$ induce a field structure on 
the set ${\mathcal Q}_\mathbb{T}$.
\end{lemma}

The proof is a routine work.
We call ${\mathcal Q}_\mathbb{T}$ {\it the field of meromorphic functions on} $\mathbb{T}$.
We simply denote by $q \in {\mathcal Q}_\mathbb{T}$ the equivalence class containing $(\Omega, q) \in Q_\mathbb{T}$.
The characteristic polynomial $f(\lambda; z)$ of the matrix $\widehat{U}(z)$
defines a polynomial $f(\lambda)$ whose coefficients are elements of $\mathcal{Q}_\mathbb{T}$.
To decompose the characteristic polynomial $f(\lambda) \in \mathcal{Q}_\mathbb{T}[\lambda]$ of $\widehat{U}$, 
we will make use of the following lemmata and proposition.

\begin{lemma}\label{lemma: no multiple root}
Let $g(\lambda) \in \mathcal{Q}_\mathbb{T}[\lambda]$ be an irreducible polynomial.
Then there exists a finite subset $\mathbb{T}_0$ of $\mathbb{T}$ such that
for every $z \in \mathbb{T} \setminus \mathbb{T}_0$, the polynomial $g(\lambda; z) \in \mathbb{C}[\lambda]$ at $z$ has no multiple root.
\end{lemma}

\begin{proof}
Note that the polynomial ring $\mathcal{Q}_\mathbb{T}[\lambda]$ is a principal ideal domain.
Choose a polynomial $h(\lambda) \in \mathcal{Q}_\mathbb{T}[\lambda]$ satisfying that
\[g(\lambda) \mathcal{Q}_\mathbb{T}[\lambda] 
+ \dfrac{\partial g}{\partial \lambda}(\lambda) \mathcal{Q}_\mathbb{T}[\lambda] 
= h(\lambda) \mathcal{Q}_\mathbb{T}[\lambda].\]
Since $\dfrac{\partial g}{\partial \lambda}(\lambda) \in h(\lambda) \mathcal{Q}_\mathbb{T}[\lambda]$, the degree of $h$ is less than that of $g$.
Since $g(\lambda)$ is irreducible, $h(\lambda)$ is an element of $\mathcal{Q}_\mathbb{T}$.
It follows that there exist 
$q_1(\lambda), q_2(\lambda) \in \mathcal{Q}_\mathbb{T}[\lambda]$ satisfying that
\[g(\lambda) q_1(\lambda) 
+ \dfrac{\partial g}{\partial \lambda}(\lambda) q_2(\lambda) 
= 1 \in \mathcal{Q}_\mathbb{T}[\lambda].\]
Choose a finite subset $\mathbb{T}_0$ of $\mathbb{T}$ such that
for every $z \in \mathbb{T} \setminus \mathbb{T}_0$, all the coefficients of $q_1(\lambda), q_2(\lambda)$ are not $\infty$ at $z$.
For such $z$, we have
\[g(\lambda; z) q_1(\lambda; z) 
+ \dfrac{\partial g}{\partial \lambda}(\lambda; z) q_2(\lambda; z) 
= 1 \in \mathbb{C}[\lambda].\]
It follows that there exists no common root $\lambda$ of $g(\lambda; z)$ and $\dfrac{\partial g}{\partial \lambda}(\lambda; z) \in \mathbb{C}[\lambda]$.
\end{proof}

\begin{proposition}\label{proposition: irreducible polynomial}
Let $g(\lambda)$ be an irreducible polynomial in $\mathcal{Q}_\mathbb{T}[\lambda]$ with degree $d$.
Suppose that the coefficient of the highest degree is $1$.
Assume that for every $z \in \mathbb{T}$, all the roots $\lambda$ of $g(\lambda; z) \in \mathbb{C}[\lambda]$ are elements of $\mathbb{T}$.
Then we have the following:
\begin{enumerate}
\item
There exists an analytic function $\lambda(\cdot) \in \mathcal{Q}_\mathbb{T}$ satisfying that for every $z \in \mathbb{T}$,
\[g(\lambda; z) 
= \prod_{\zeta \colon \zeta^d = z} (\lambda - \lambda(\zeta))\]
\item\label{item: rotation}
Assume that
another analytic function $\widetilde{\lambda}(\cdot)$ defined around $\mathbb{T}$ satisfies
\[g(\lambda; z) 
= \prod_{\zeta \colon \zeta^d = z} 
\left( \lambda - \widetilde{\lambda}(\zeta) \right).\]
Then there exists a natural number $c$ satisfying that
\begin{eqnarray*}
c &\in& \{0, 1, 2, \cdots, d - 1 \},\\
\widetilde{\lambda}(\zeta) &=& \lambda( \exp(2 \pi i c / d) \zeta), 
\quad \zeta \in \mathbb{T}.
\end{eqnarray*}
\end{enumerate}
\end{proposition}

\begin{proof}
Since the roots $\lambda$ of $g(\lambda; z)$ uniformly bounded on $z \in \mathbb{T}$,
there exists a domain $\Omega$ containing $\mathbb{T}$ on which the coefficients of $g(\lambda)$ have no poles.
By the same argument as the book \cite[Chapter 8, Section 2]{Ahlfors} by Ahlfors, the set of germs whose graphs are included in
\[\mathcal{G} = \{(z, \lambda) \in \Omega \times \mathbb{C} \ |\ g(\lambda; z) = 0 \} \]
gives a Riemann surface.
In the book by Ahlfors, it is proved that an irreducible polynomial defines a (compact) Riemann surface.
In our argument, the coefficients of $\lambda^d, \lambda^{d-1}, \cdots, 1$ are not necessarily polynomials of $z$ but analytic functions of $z$.
However, the argument by Ahlfors works in our framework and shows that $\mathcal{G}$ gives a (not necessarily compact) Riemann surface.

We prove that the Riemann surface has no branch point on $\mathbb{T}$.

Let $z_0$ be an arbitrary element of $\mathbb{T}$. 
Let $D \subset \Omega$ be a tiny open disc including $z_0$ such that there exists no branch point in $\overline{D} \setminus \{z_0\}$ and that the circle $\partial D$ intersects with $\mathbb{T}$ at right angles.
Let $\{z_1, z_2\}$ be the intersection of $\mathbb{T}$ and the boundary $\partial D$. 
Pick up an analytic germ 
$z \mapsto \lambda(z)$ defined around $z_1$ whose graph is included in
$\mathcal{G}$.
Denote by $C_\mathrm{out}$ the path such that the starting point is $z_1$ and the terminal point is $z_2$ and that $C_\mathrm{out}$ goes outside of $\mathbb{T}$. 
Denote by $C_\mathrm{in}$ the path such that the staring point is $z_1$ and the terminal point is $z_2$ and that $C_\mathrm{in}$ goes inside of $\mathbb{T}$. 
Let $\lambda_\mathrm{out}(z)$ be the germ defined around $z_2$ given by the analytic continuation of $\lambda(z)$ along $C_\mathrm{out}$.
Let $\lambda_\mathrm{in}(z)$ be the germ defined around $z_2$ given by the analytic continuation of $\lambda(z)$ along $C_\mathrm{in}$.

Because the germ $\lambda$ satisfies $g(\lambda(z), z) = 0$, by the identity theorem,
we have $g(\lambda_\mathrm{in}(z), z) = 0$.
By assumption, for $z \in \mathbb{T}$, the absolute values of the roots of $g(\lambda, z)$ are $1$.
It follows that 
if $z$ is in $\mathbb{T}$ and close to $z_2$, then
\begin{eqnarray*}
|\lambda_\mathrm{in}(z)| &=& 1,\\
\dfrac{1}{\ \overline{\lambda_\mathrm{in}(z)}\ } &=& \lambda_\textrm{in} (z).
\end{eqnarray*}

For $z \in \mathbb{T}$ close to $z_1$, we have $|\lambda(z)| = 1$. For such $z$, we have
\begin{eqnarray*}
\dfrac{1}{\ \overline{z}\ } &=& z,\\
\lambda\left(\dfrac{1}{\ \overline{z}\ }\right) &=& \lambda(z) = \dfrac{1}{\ \overline{\lambda(z)}\ }.
\end{eqnarray*}
Let us move $z$ from $z_1$ along $C_{\rm in}$.
Then $\dfrac{1}{\ \overline z\ }$ moves from $z_1$ along $C_{\rm out}$.
By Schwarz reflection principle, 
for $z \in \Omega$ around $z_2$,
we have
\[\lambda_\mathrm{out}\left(\frac{1}{\ \overline{z}\ }\right) = \frac{1}{\ \overline{\lambda_\mathrm{in}(z)}\ }.\]
If $z$ is in $\mathbb{T}$ and close to $z_2$, then we have
\[\lambda_\mathrm{out}(z) = \dfrac{1}{\ \overline{\lambda_\mathrm{in}(z)}\ }.\]
It follows that $\lambda_\mathrm{out}(z) = \lambda_\mathrm{in}(z)$.
By the identity theorem, this equality holds on a neighborhood of $z_2$.
It means that two germs given by $\lambda_\mathrm{out}$ and $\lambda_\mathrm{in}$ at $z_2$ are identical.
We conclude that the analytic continuation of $\lambda(\cdot)$ on the tiny circle $\partial D$ is unique and that $z_0$ is not a branch point.
It follows that there exists no branch point on $\mathbb{T}$.

Let us take an analytic germ $\lambda_0(z)$ at $1$ whose graph is included in $\mathcal{G}$.
Since the Riemann surface has no branch point on $\mathbb{T}$, there exists an analytic continuation 
\[\mathbb{R} \ni \theta \mapsto \lambda_\theta,\]
where $\lambda_\theta$ is a germ at $e^{i \theta}$.
By the finiteness of the roots, there exists a natural number $c$ such that $2 \pi c$ is the period of the analytic continuation.
Define an analytic function $\lambda(\cdot)$ on $\mathbb{T}$ by
\[\lambda(\zeta) = \lambda_{c \arg \zeta}(\zeta^c).\]
Note that for every $c$-th root $\zeta$ of $z$, we have
\[g(\lambda(\zeta), z) = 0.\]
For every $1 \le b \le c$, the function
\[z \mapsto b \textrm{-th\ elementary\ symmetric\ polynomial\ of\ } 
\{\lambda(\zeta) \ |\ \zeta^c = z\}. \]
defines an analytic function of $z$.
It follows that
\[g_1(\lambda; z) := \prod_{\zeta \colon \zeta^c = z} (\lambda - \lambda(\zeta))\]
gives an element $g_1(\lambda)$ of $\mathcal{Q}_\mathbb{T}[\lambda]$.

We next prove that $g_1(\lambda)$ is identical to $g(\lambda)$.
By the definition of $\lambda(\zeta)$,
for every $z \in \mathbb{T}$, every root of 
$g_1(\lambda; z)$
is that of $g(\lambda; z)$.
By Lemma \ref{lemma: no multiple root},
for almost every $z \in \mathbb{T}$, $g(\lambda; z) \in \mathbb{C}[\lambda]$ has mutually different $d$ roots.
It follows that for such $z$, the complex numbers $\{\lambda(\zeta) \ |\ \zeta^c = z\}$ are mutually different.
Consider the remainder $r(\lambda)$ obtained by the polynomial long division 
\[g(\lambda) = q(\lambda) g_1(\lambda) + r(\lambda) \in \mathcal{Q}_\mathbb{T}[\lambda].\]
Since $g_1(\lambda)$ is monic, all the coefficients of $q(\lambda)$ and $r(\lambda)$ are realized by complex-valued analytic functions.
Substituting $z$, we obtain the identity
\[g(\lambda; z) = q(\lambda; z) g_1(\lambda; z) + r(\lambda; z) \in \mathbb{C}[\lambda].\]
For almost every $z \in \mathbb{T}$, $r(\lambda; z) \in \mathbb{C}[\lambda]$ has mutually different $c$ roots $\{\lambda(\zeta) \ |\ \zeta^c = z\}$.
Since the degree of $r(\lambda; z) \in \mathbb{C}[\lambda]$ is less than $c$, we have $r(\lambda; z) = 0$ for such $z$.
By continuity, for every $z \in \mathbb{T}$, $r(\lambda; z) = 0$.
Therefore we have
\[g(\lambda) = q(\lambda) g_1(\lambda) \in \mathcal{Q}_\mathbb{T}[\lambda].\]
Since $g(\lambda)$ is irreducible, two polynomials $g$ and $g_1$ are identical.

For the second item of the theorem, assume that there exists an analytic function 
$\widetilde{\lambda}$ on $\mathbb{T}$ such that
\[g(\lambda; z) 
= \prod_{\zeta \colon \zeta^d = z} (\lambda - \lambda(\zeta))
= \prod_{\zeta \colon \zeta^d = z} 
\left( \lambda - \widetilde{\lambda}(\zeta) \right).\]
By Lemma \ref{lemma: no multiple root}, there exists $z_0 \in \mathbb{T}$ such that if $z \in \mathbb{T}$ is close to $z_0$, then 
$\{\lambda(\zeta) \ |\ \zeta^d = z\}$
consists of $d$ elements. 
Choose a $d$-th root $\zeta_0$ of $z_0$.
Then there exists $c \in \{0, 1, \cdots, d - 1 \}$ such that $\widetilde{\lambda}(\zeta_0) =  \lambda( \exp(2 \pi i c / d) \zeta_0)$.
By continuity of $\lambda$ and of $\widetilde{\lambda}$, if $\zeta \in \mathbb{T}$ is close to $\zeta_0$, then
$\widetilde{\lambda}(\zeta) =  \lambda( \exp(2 \pi i c / d) \zeta)$.
By the identity theorem, the equality holds for every $\zeta \in \mathbb{T}$.
\end{proof}

\begin{remark}
In the case that the quantum walk $U$ has finite propagation,
the characteristic polynomial $f(\lambda; z)$ of $\widehat{U}(z)$ is a polynomial of $\lambda$,
 $z$, $z^{-1}$.
Let $g(\lambda) \in \mathbb{C}[z, z^{-1}] [\lambda]$ be an irreducible factor of 
$f(\lambda) \in \mathbb{C}[z, z^{-1}] [\lambda]$.
The Riemann surface given by $g$ is a {\it compact} Riemann surface.
This is not used for the rest of this paper, but this is interesting.
\end{remark}

\begin{lemma}\label{lemma: iff reducible}
Let $\lambda(\cdot)$ be an analytic function defined on $\mathbb{T}$.
Define $g(\lambda) \in \mathcal{Q}_\mathbb{T}[\lambda]$ by
\[g(\lambda; z) 
= \prod_{\zeta \colon \zeta^d = z} (\lambda - \lambda(\zeta)).\]
The following two conditions are equivalent:
\begin{enumerate}
\item
The polynomial $g(\lambda) \in \mathcal{Q}_\mathbb{T}[\lambda]$ is reducible.
\item
There exists a natural number $c \in \{1, 2, \cdots, d - 1\}$ satisfying that
\begin{eqnarray*}
\lambda \left( \exp\left( \dfrac{2 \pi i c}{d} \right) \zeta \right) &=& \lambda(\zeta),
\end{eqnarray*}
for every $\zeta \in \mathbb{T}$.
\end{enumerate}
If the above conditions hold ture,
then there exist natural numbers $b, c$ and an analytic map $\widetilde{\lambda} \colon \mathbb{T} \to \mathbb{T}$ satisfying that $b c = d$ and that
\begin{eqnarray*}
\widetilde{\lambda}(\zeta^b) &=& \lambda(\zeta), \quad \zeta \in \mathbb{T},\\
g(\lambda; z) 
&=& \left( \prod_{\eta \colon \eta^c = z} 
\left( \lambda - \widetilde\lambda(\eta) \right) \right)^b.
\end{eqnarray*}

\end{lemma}

\begin{proof}
Suppose that there exists $c \in \{1, 2, \cdots, d - 1\}$ satisfying that
\begin{eqnarray*}
\lambda( \exp(2 \pi i c / d) \zeta) &=& \lambda(\zeta).
\end{eqnarray*}
Choose minimum value of $c$ satisfying the above property. Such a natural number $c$ divides $d$.
Define $b$ by $d / c$.
Define an analytic function $\widetilde{\lambda} \colon \mathbb{T} \to \mathbb{T}$
by
\[\widetilde{\lambda}(\eta) = \lambda(b\textrm{-th\ root\ of\ }\eta).\]
Since $\lambda( \exp(2 \pi i / b) \zeta) = \lambda(\zeta)$,
$\widetilde{\lambda}$ is well-defined.
The polynomial
\[
\left(
\prod_{\eta \colon \eta^c = z} \left( \lambda - \widetilde{\lambda}(\eta) \right) 
\right)^b \in \mathbb{C}[\lambda]\]
is identical to $g(\lambda; z)$.
It follows that $g(\lambda)$ is not irreducible.

Suppose that $g(\lambda)$ is not irreducible.
Take an irreducible monic polynomial 
$g_1(\lambda) \in \mathcal{Q}_\mathbb{T}[\lambda]$ which divides $g(\lambda)$.
Let $c$ be the degree of $g_1(\lambda)$.
By Proposition \ref{proposition: irreducible polynomial},
there exist a natural number $c$ and an analytic function $\widetilde{\lambda}$
such that
\[g_1(\lambda) = 
\prod_{\eta \colon \eta^c = z} \left( \lambda - \widetilde{\lambda}(\eta) \right).\]
Because
$\prod_{\eta \colon \eta^c = z} (\lambda - \widetilde{\lambda}(\eta))$
divides
$\prod_{\zeta \colon \zeta^d = z} (\lambda - \lambda(\zeta))$,
the germ of $\widetilde{\lambda}$ around $1$ is realized by a germ of $\lambda$ around some $d$-th root $\zeta_0$ of $1$.
More precisely,
if $\eta$ is close to $1$, if $\zeta$ is close to $\zeta_0$, and if $\eta^c = \zeta^d$, 
then
$\widetilde{\lambda}(\eta) = \lambda(\zeta)$.
Let us move $\eta$ on $\mathbb{T}$ in the anticlockwise direction. 
Under the condition that $\eta^c = \zeta^d$, as $\arg \eta$ moves from $0$ to $2 \pi$, $\arg \zeta$ moves from $\arg \zeta_0$ to $\arg \zeta_0 + 2 \pi c / d$.
We have $\widetilde{\lambda} (\exp(2 \pi i) \eta) = \lambda (\exp(2 \pi i c/d) \zeta)$.
Since $\widetilde{\lambda} (\exp(2 \pi i) \eta) = \widetilde{\lambda} (\eta)$, we have 
\[\lambda( \exp(2 \pi i c / d) \zeta) = \lambda(\zeta),\]
for $\zeta \in \mathbb{T}$ close to $1$.
By the identity theorem, for every $\zeta \in \mathbb{T}$, the equation holds.
\end{proof}

\subsection{Eigenvalue function for a quantum walk}
\label{subsection: eigenvalue functions}
The characteristic polynomial 
\[f(\lambda; z) 
= \det \left( \left( \lambda \delta_{k, l} - \widehat{U}(z ; k, l) \right)_{k, l} \right)\] 
of the inverse Fourier transform $\widehat{U}(z)$
induces a polynomial $f(\lambda) \in \mathcal{Q}_\mathbb{T}[\lambda]$. The coefficients of $f(\lambda)$ are analytic functions defined on $\mathbb{T}$. The polynomial admits a decomposition into irreducible polynomials, and each irreducible factor admits such an expression as in Proposition \ref{proposition: irreducible polynomial}.
Thus we have the following proposition:

\begin{proposition}\label{proposition: eigenvalue functions}
There exist 
\begin{itemize}
\item
a sequence of natural numbers
$d(1), d(2), \cdots, d(m) $
whose sum is $n$,
\item
analytic functions $\lambda_1, \cdots, \lambda_m \colon \mathbb{T} \to \mathbb{T}$,
\end{itemize}
satisfying that the characteristic polynomial $f(\lambda; z)$ of $\widehat{U}(z)$ is given by
\begin{eqnarray*}
f(\lambda; z) =
\prod_{j = 1}^m \prod_{\zeta \colon \zeta^{d(j)} = z} \left( \lambda - \lambda_j(\zeta) \right).
\end{eqnarray*}
\end{proposition}

In our argument before Proposition \ref{proposition: eigenvalue functions},
the factor $\prod_{\zeta \colon \zeta^{d(j)} = z} \left( \lambda - \lambda_j(\zeta) \right)$
is irreducible.
However the equation in Proposition \ref{proposition: eigenvalue functions}
does not imply each factor is irreducible, because the factor may admit further decomposition as in Lemma \ref{lemma: iff reducible}.

\begin{definition}\label{definition: system of eigenvalue functions}
For the quantum walk $U$, the $m$-tuple of pairs
\[\left( (d(1), \lambda_1), (d(2), \lambda_2), \cdots, (d(m), \lambda_m) \right)\]
satisfying the equation in Proposition \ref{proposition: eigenvalue functions} is called 
{\rm a system of eigenvalue functions} of
$\widehat{U}$.
\end{definition}

For the quantum walk $U$, the natural number $m$ and the system of eigenvalue functions of $\widehat{U}$ are not necessarily unique.
It admits the following three types of replacements:
\begin{enumerate}
\item
Permutation on the index $\{1, 2, \cdots, m\}$.
\item
Rotation on the function $\lambda_j$.
More precisely, The system admits the replacement of $\lambda_j(\zeta)$ with
\[\lambda_j \left( \exp \left( \frac{2 \pi i c}{d(j)} \right) \zeta \right),\] 
where $c$ is a natural number.
See Proposition \ref{proposition: irreducible polynomial} (\ref{item: rotation}).
\item
Decomposition described in Lemma \ref{lemma: iff reducible}.
More precisely,
in the case that 
\[\lambda_j(\zeta)
=
\lambda_j \left( \exp \left( \frac{2 \pi i c}{d(j)} \right) \zeta \right),
\quad
b := \dfrac{d(j)}{c} \in \mathbb{N},\]
the pair $(d(j), \lambda_j)$ can be replaced with the $b$-tuple of pairs
\[\left(c, \widetilde\lambda\right), \left(c, \widetilde\lambda\right), 
\cdots, \left(c, \widetilde\lambda\right).\]
The new eigenvalue function is given by
\[\widetilde\lambda(\eta) = \lambda_j(\textrm{the\ } b \textrm{-th\ root\ of\ } \eta)\]
\end{enumerate}

If the third procedure can not be applied to the system, the system is said to be {\it indecomposable}.
If two systems of eigenvalue functions are given, by applying the above procedures (1), (2), and (3),
we obtain a common indecomposable system of eigenvalue functions.
If two indecomposable systems of eigenvalue functions are given, by applying the procedures (1), (2) to one system,
we obtain the other system.
This is a conclusion of the uniqueness of the irreducible decomposition of $f(\lambda)$.

\begin{definition}\label{definition: absolute value of winding number}
Let 
\[\left( (d(1), \lambda_1), (d(2), \lambda_2), \cdots, (d(m), \lambda_m) \right)\]
be a system of eigenvalue functions of $\widehat{U}$.
We denote by $w(\lambda_j)$ the winding number of the analytic map $\lambda_j \colon \mathbb{T} \to \mathbb{T}$.
We define a quantity $|w|(U)$ by
the sum
\[\sum_{j = 1}^m |w(\lambda_j)|.\]
\end{definition}
The quantity $|w|(U)$ is uniquely determined by $U$, 
because the sum is preserved under the procedures (1), (2), and (3).

\subsection{Analytic section of eigenvectors}
\label{subsection: analytic section}

The following is 
a structure theorem on the inverse Fourier transform $\widehat{U}$ of an analytic homogeneous quantum walk $U$ on $\mathbb{Z}$
\begin{proposition}\label{proposition: structure of dual}
For every indecomposable system of eigenvalue functions
\[\left( (d(1), \lambda_1), (d(2), \lambda_2), \cdots, (d(m), \lambda_m) \right)\]
of $\widehat{U}$,
there exist 
analytic maps
${\bf v}_1, \cdots, {\bf v}_m \colon \mathbb{T} \to \mathbb{C}^n$
satisfying the following:
\begin{itemize}
\item
for every $z \in \mathbb{T}$, 
\[
\left\{ {\bf v}_j (\zeta) \ 
\left| \ 1 \le j \le m, 
\zeta \in \mathbb{T}, \zeta^{d(j)} = z \right.\right\} \]
forms an orthonormal basis of $\mathbb{C}^n$,
\item
for every $1 \le j \le m$, and for every $\zeta \in \mathbb{T}$,
\[\widehat{U}\left( \zeta^{d(j)} \right) {\bf v}_j(\zeta) = \lambda_j(\zeta) {\bf v}_j(\zeta).\]
\end{itemize}
\end{proposition}


\begin{proof}
Because the indecomposable system of the eigenvalue functions is essentially unique,
it suffices to construct the required analytic sections of eigenvectors for some indecomposable system.

Let
$f(\lambda) = g_1(\lambda) g_2(\lambda) \cdots g_m(\lambda)$
be an irreducible decomposition of the characteristic polynomial $f(\lambda) \in \mathcal{Q}_\mathbb{T}[\lambda]$.
For $1 \le j \le m$, denote by $d(j)$ the degree of $g_j(\lambda)$.
We may assume 
that $g_1, g_2, \cdots, g_m$ are monic.
By Proposition \ref{proposition: irreducible polynomial},
there exists an analytic function $\lambda_j \colon \mathbb{T} \to \mathbb{T}$ satisfying that 
\[g_j(\lambda; z) = \prod_{\zeta \colon \zeta^{d(j)} = z} (\lambda - \lambda_j(\zeta)). \]
The collection of such $\lambda_j(\zeta)$ is the set of all the roots of $f(\lambda; z)$.
We may further assume that
\[g_1 = \cdots = g_p, g_1 \neq g_{p + 1}, g_1 \neq g_{p + 2}, \cdots, g_1 \neq g_m.\]
and that
$\lambda_1 = \cdots = \lambda_p$.
By Lemma \ref{lemma: iff reducible}, on a coset of finite subset of $\mathbb{T}$,
the eigenvalue $\lambda_1(\zeta)$ of $\widehat{U}(\zeta^{d(1)})$ is different from other eigenvalues
\begin{eqnarray*}
&\lambda_1 \left( \exp \left( \dfrac{2 \pi k i}{d(1)} \right) \zeta \right),&
\quad k = 1, 2, \cdots, d(1) - 1,\\
&\lambda_j (\eta),&
\quad j = p + 1, p + 2, \cdots, m, \quad \eta \textrm{\ is a\ } d(j) \textrm{-th root of\ } \zeta^{d(1)}. 
\end{eqnarray*}

Define a matrix $X^{(1)}(\zeta) \in M_n(\mathbb{C})$ by the product
\[X^{(1)}(\zeta) = 
\prod_{\eta \colon \eta^{d(1)} = \zeta^{d(1)}, \eta \neq \zeta} 
\left( \lambda_1(\eta) - \widehat{U}\left( \zeta^{d(1)} \right) \right)
\cdot \prod_{j = p + 1}^m g_j \left( \widehat{U}\left( \zeta^{d(1)} \right) \right). \]
Note that for every $\zeta \in \mathbb{T}$, $\widehat{U}\left( \zeta^{d(1)} \right)$ is diagonalizable and satisfies
\[
\left(\lambda_1(\zeta) - \widehat{U}\left( \zeta^{d(1)} \right) \right) X^{(1)} (\zeta)
= f \left( \widehat{U} \left( \zeta^{d(1)} \right); \zeta^{d(1)} \right)
= O.
\]
On a coset of a finite subset of $\mathbb{T}$, the image of $X^{(1)}(\zeta)$ is the eigenspace of $\widehat{U}\left( \zeta^{d(1)} \right)$
whose eigenvalue is $\lambda_1(\zeta)$.
On a coset of a finite subset of $\mathbb{T}$, the rank of $X^{(1)}(\zeta)$ is equal to the multiplicity $p$ of the eigenvalue $\lambda_1(\zeta)$.
Let $X^{(1)}_l(\zeta)$ be the $l$-th column of $X^{(1)}(\zeta)$.
Note that for  every choice of  $1 \le l(1) < l(2) < \cdots < l(p) \le n$, 
the map 
\[\zeta \mapsto X^{(1)}_{l(1)}(\zeta) \wedge \cdots \wedge X^{(1)}_{l(p)}(\zeta) \in \wedge^{p} \mathbb{C}^n\]
to the $p$-th exterior product
is analytic.
There exists a collection $l(1) < l(2) < \cdots < l(p)$ of labels satisfying that
\[X^{(1)}_{l(1)}(\zeta) \wedge \cdots \wedge X^{(1)}_{l(p)}(\zeta) \in \wedge^p \mathbb{C}^n\]
is not the constant map $\bf 0$ and the zero set is at most finite.
This means that 
the set of vectors
\[ \left\{X^{(1)}_{l(1)}(\zeta), \cdots, X^{(1)}_{l(p)}(\zeta) \right\} \]
is linearly independent and spans the image of $X^{(1)}(\zeta)$ on a coset of a finite subset of $\mathbb{T}$.

By Lemma \ref{lemma: orthonormal basis},
there exists analytic maps
${\bf v}^{(1)}, {\bf v}^{(2)}, \cdots, {\bf v}^{(p)} \colon \mathbb{T} \to \mathbb{C}^n$ satisfying the following conditions:
\begin{itemize}
\item
for every $\zeta \in \mathbb{T}$, $\left\{ {\bf v}^{(1)}(\zeta), {\bf v}^{(2)}(\zeta), \cdots, {\bf v}^{(p)}(\zeta) \right\}$ forms an orthonormal system,
\item
on a coset of a finite subset of $\mathbb{T}$,  the linear span of 
$\left\{ X^{(1)}_{l(1)}(\zeta), \cdots, X^{(1)}_{l(p)}(\zeta) \right\}$ is identical to the linear span of $\left\{ {\bf v}^{(1)}(\zeta), {\bf v}^{(2)}(\zeta), \cdots, {\bf v}^{(p)}(\zeta) \right\}$.
\end{itemize}
It follows that on a coset of a finite subset of $\mathbb{T}$, 
the following four subspaces of $\mathbb{C}^n$ coincide:
\begin{itemize}
\item
the eigenspace of $\widehat{U}\left( \zeta^{d(1)} \right)$
whose eigenvalue is $\lambda_1(\zeta)$,
\item
the image of $X^{(1)}(\zeta)$,
\item
the linear span of $\left\{X^{(1)}_{l(1)}(\zeta), \cdots, X^{(1)}_{l(p)}(\zeta) \right\}$,
\item
the linear span of
$\left\{ {\bf v}^{(1)}(\zeta), {\bf v}^{(2)}(\zeta), \cdots, {\bf v}^{(p)}(\zeta) \right\}$.
\end{itemize}
The vectors
${\bf v}^{(1)}(\zeta), {\bf v}^{(2)}(\zeta), \cdots, {\bf v}^{(p)}(\zeta)$ are eigenvectors of $\widehat{U}\left( \zeta^{d(1)} \right)$ whose eigenvalue is $\lambda_1(\zeta)$, on the coset of $\mathbb{T}$.
By the continuity of ${\bf v}^{(1)}(\zeta)$, ${\bf v}^{(2)}(\zeta)$, $\cdots$, ${\bf v}^{(p)}(\zeta)$ and $\lambda_1(\zeta)$, it turns out that there exists no exception. 

By Lemma \ref{lemma: iff reducible},
every two eigenvalue functions chosen from
\[\lambda_1(\zeta), \lambda_1(\exp(2 \pi i/d(1)) \zeta), \cdots, 
\lambda_1(\exp(2 \pi i (d(1) -1)/d(1)) \zeta)\]
are not identical.
With finite exceptions, for fixed $z \in \mathbb{T}$, the $(p \times d(1))$ vectors
\[ \left\{ {\bf v}^{(j)}(\zeta)\ |\ 1 \le j \le p, \zeta^{d(1)} = z \right\},\]
form an orthonormal system in $\mathbb{C}^n$.
Again by continuity, the exceptions are removed.

Now we proceed to the next step.
Rearranging the index, we may further assume that
\[g_{p + 1} = \cdots = g_{p + q}, g_{p +1} \neq g_{p + q + 1}, \cdots,\]
and that $\lambda_{p + 1} = \cdots = \lambda_{p + q}$.
Define $X^{(p + 1)}(\zeta)$ by
\begin{eqnarray*}
X^{(p + 1)}(\zeta)
&=& 
g_1 \left( \widehat{U}(\zeta^{d(p + 1)}) \right)
\cdot \prod_{\eta \colon \eta^{d(p + 1)} = \zeta^{d(p + 1)}, \eta \neq \zeta} 
\left( \lambda_{p + 1}(\eta) - \widehat{U}(\zeta^{d(p + 1)}) \right)\\
& & \hspace{60mm}
\cdot \prod_{j = p + q + 1}^m g_j \left( \widehat{U}(\zeta^{d(p + 1)}) \right). 
\end{eqnarray*}
Using column vectors of $X^{(p + 1)}(\zeta)$,
we can construct a section of orthonormal basis
\[{\bf v}^{(p + 1)}(\zeta), {\bf v}^{(p + 2)}(\zeta), \cdots, {\bf v}^{(p + q)}(\zeta)\]
which consists of eigenvectors of $\widehat{U} \left( \zeta^{d(p + 1)} \right)$ whose eigenvalues are
$\lambda_{p + 1}(\zeta)$. 
For every $z \in \mathbb{T}$, $(q \times d(p + 1))$ vectors
\[\left\{ {\bf v}^{(j)}(\zeta)\ |\ p + 1 \le j \le p + q, \zeta^{d(p + 1)} = z \right\},\]
forms an orthonormal system. 
On a coset of a finite subset of $\mathbb{T}$,
the roots of $g_1(\lambda; z)$ are different from those of $g_{p + 1}(\lambda; z)$.
Therefore the members of the system are perpendicular to
\[\left\{ {\bf v}^{(j)}(\zeta)\ |\ 1 \le j \le p, \zeta^{d(1)} = z \right\},\]
on the coset. Again by continuity, it turns out that there exists no exception.

Repeating this procedure, we finish the construction of ${\bf v}^{(j)} (\zeta)$, $1 \le j \le n$.
\end{proof}

\section{Realization by continuous-time QW}
\label{section: DTQW by CTQW}

In this section, we first construct a collection of typical quantum walks, 
which is called {\it model quantum walks}.
These walks are like atoms in the world of discrete-time analytic homogeneous quantum walks $U$ on $\mathbb{Z}$.
Such a walk $U$ is equivalent to a direct sum of model quantum walks.

\subsection{Model quantum walks}\label{subsection: model quantum walk}

We introduce the model quantum walk $U_{d, \lambda}$, which is constructed by a natural number $d$ and
an analytic function $\lambda \colon \mathbb{T} \to \mathbb{T}$.
Let
\[\lambda(\zeta) = \sum_{s = -\infty}^\infty c(s) \zeta^s 
\]
be the Laurent series of $\lambda(\zeta)$. For $k, l \in \{1, 2, \cdots, d\}$, define an analytic
operator
$U_{k, l}$ acting on $\ell_2(\mathbb{Z})$ by
\begin{eqnarray*}
U_{k, l}
= \sum_{s = -\infty}^\infty c(k - l + d s) S_s
\end{eqnarray*}
Define an analytic operator $U_{d, \lambda}$ acting on $\ell_2(\mathbb{Z}) \otimes \mathbb{C}^d$ by
\[ U_{d, \lambda} = \left( U_{k, l} \right)_{k, l}.\]
Let $\lambda_{k, l} \colon \mathbb{T} \to \mathbb{C}$ be the function
defined by
\begin{eqnarray*}
\lambda_{k, l}(z) &=& \sum_{s = -\infty}^\infty c(k - l + d s) z^s.
\end{eqnarray*}
The inverse Fourier transform $\widehat{U_{k, l}} = \mathcal{F}^{-1} U_{k, l} \mathcal{F}$
is identical to the multiplication operator $M[\lambda_{k, l}]$ by $\lambda_{k, l}$. 

In the case that $d = 1$, $\widehat{U_{1, \lambda}}$ is nothing other than the multiplication operator $M[\lambda]$ by the function $\lambda$.
The unitary $U_{1, \lambda}$ is the operator given by the Fourier transform of $\lambda$.
The operator is expressed by 
\[U_{1, \lambda} = \sum_{s= -\infty}^\infty c(s) S_s.\]

We first prove that the natural number $d$ does not have an important role.
Define a unitary operator $W_d \colon
\ell_2 (\mathbb{Z}) \otimes \mathbb{C}^d \to \ell_2 (\mathbb{Z})$ by
\[W_d (\delta_s \otimes \delta_k) = \delta_{k + d s}, \quad s \in \mathbb{Z}, k \in \{1, 2, \cdots, d\}.\]
We call $W_d$ the {\it rearrangement}.

\begin{lemma}\label{lemma: rearrangement}
Let $\lambda \colon \mathbb{T} \to \mathbb{T}$ be an analytic function and let $d$ be a natural number.
Then we have
$U_{d, \lambda} = W_d^* U_{1, \lambda} W_d.$
\end{lemma}

\begin{proof}
Fix arbitrary $t \in \mathbb{Z}$ and $l \in \{1, 2, \cdots, d\}$ for a while.
We hit the vectors $\delta_t \otimes \delta_l$
to the unitary operators $W_d U_{d, \lambda}$ and $U_{1, \lambda} W_d$.
We obtain the following equation:
\begin{eqnarray*}
W_d U_{d, \lambda} (\delta_t \otimes \delta_l)
&=&
W_d \sum_{k = 1}^d \sum_{s = -\infty}^\infty c(k - l + d s) S_s \delta_t \otimes \delta_k
\\ 
&=&
W_d \sum_{k = 1}^d \sum_{s = -\infty}^\infty c(k - l + d s) \delta_{s + t} \otimes \delta_k
\\ 
&=&
\sum_{k = 1}^d \sum_{s = -\infty}^\infty c(k - l + d s) \delta_{k + d(s + t)}.
\end{eqnarray*}
Every integer $\sigma$ is uniquely expressed by $\sigma = k - l  + d s, 
k \in \{1, \cdots, d\}, s \in \mathbb{Z}$.
We get the equation
\begin{eqnarray*}
W_d U_{d, \lambda} (\delta_t \otimes \delta_l)
&=&
\sum_{\sigma = -\infty}^\infty c(\sigma) \delta_{l + \sigma + d t}.
\end{eqnarray*}
We also have
\begin{eqnarray*}
U_{1, \lambda} W_d (\delta_t \otimes \delta_l)
=
\sum_{s = -\infty}^\infty c(s) S_s \delta_{l + d t}
=
\sum_{s = -\infty}^\infty c(s) \delta_{l + s + d t}.
\end{eqnarray*}
They are identical.
\end{proof}

\begin{lemma}\label{lemma: calculation of the model quantum walks}
For analytic maps $\lambda, \lambda_1, \lambda_2 \colon \mathbb{T} \to \mathbb{T}$ and a natural number $d$,
\[U_{d, \lambda}^* = U_{d, \overline{\lambda}}, \quad 
U_{d, \lambda_1} U_{d, \lambda_2} = U_{d, \lambda_1 \lambda_2}.\]
\end{lemma}

\begin{proof}
Using Lemma \ref{lemma: rearrangement}, we have
\begin{eqnarray*}
U_{d, \lambda}^* 
&=& W_d^* U_{1, \lambda}^* W_d
= W_d^* (\mathcal{F} M[\lambda] \mathcal{F}^{-1})^* W_d\\
&=& W_d^* \mathcal{F} M \left[ \overline{\lambda} \right] \mathcal{F}^{-1} W_d
= W_d^* U_{1, \overline{\lambda}} W_d\\
&=& U_{d, \overline{\lambda}}.
\end{eqnarray*}
We also have
\begin{eqnarray*}
U_{d, \lambda_1} U_{d, \lambda_2}  
&=& W_d^* U_{1, \lambda_1} U_{1, \lambda_2} W_d
= W_d^* \mathcal{F} M[\lambda_1] M[\lambda_2] \mathcal{F}^{-1} W_d\\
&=& W_d^* \mathcal{F} M[\lambda_1 \lambda_2] \mathcal{F}^{-1} W_d
= W_d^* U_{1, \lambda_1 \lambda_2} W_d\\
&=& U_{d, \lambda_1 \lambda_2}.  
\end{eqnarray*}
\end{proof}

\begin{lemma}
Let $\lambda \colon \mathbb{T} \to \mathbb{T}$ be an analytic function and let $d$ be a natural number.
For every $z \in \mathbb{T}$, the operator $U_{d, \lambda}$ is unitary.
\end{lemma}

\begin{proof}
By Lemma \ref{lemma: calculation of the model quantum walks}, we have
\begin{eqnarray*}
U_{d, \lambda}^* U_{d, \lambda}
=
U_{d, \overline{\lambda}} U_{d, \lambda} 
=
U_{d, \overline{\lambda} \lambda}
=
U_{d, 1}
=
1. 
\end{eqnarray*}
The operator
$U_{d, \lambda}^* U_{d, \lambda}$ is also the identity operator $1$.
\end{proof}

We calculate the eigenvalue functions of the model quantum walks.

\begin{lemma}\label{lemma: eigenvectors for model}
For every $\zeta \in \mathbb{T}$, the column vector
\[\left( 1, \zeta^{-1}, \zeta^{-2}, \cdots, \zeta^{1 - d} \right)^\mathrm{T}\]
is an eigenvector of $\widehat{U_{d, \lambda}}(\zeta^d)$ whose eigenvalue is $\lambda(\zeta)$.
\end{lemma}

\begin{proof}
We directly compute.
The $k$-th entry of the vector
\[ \widehat{U_{d, \lambda}}(\zeta^d) 
\cdot 
(1, \zeta^{-1}, \zeta^{-2}, \cdots, \zeta^{1 - d})^\mathrm{T} \]
is
\begin{eqnarray*}
\sum_{l = 1}^d \lambda_{k, l}(\zeta^d) \zeta^{1 - l}
&=&
\sum_{l = 1}^d \sum_{s = -\infty}^\infty c(k - l + d s) \zeta^{1 - l + d s}\\
&=&
\zeta^{1 - k} \sum_{l = 1}^d \sum_{s = -\infty}^\infty c(k - l + d s) \zeta^{k - l + d s}.
\end{eqnarray*}
Every integer $\sigma$ is uniquely expressed by $k - l + d s$, $l \in \{1, \cdots, d\}$, $s \in \mathbb{Z}$.
It follows that
\begin{eqnarray*}
\sum_{l = 1}^d \lambda_{k, l}(\zeta^d) \zeta^{1 - l}
=
\zeta^{1 - k} \sum_{\sigma = -\infty}^\infty c(\sigma) \zeta^\sigma
=
\lambda(\zeta) \zeta^{1 - k}.
\end{eqnarray*}
We obtain the following equation:
\[ \widehat{U_{d, \lambda}}(\zeta^d) \cdot (1, \zeta^{-1}, \zeta^{-2}, \cdots, \zeta^{1 - d})^\mathrm{T} 
=
\lambda(\zeta) \cdot (1, \zeta^{-1}, \zeta^{-2}, \cdots, \zeta^{1 - d})^\mathrm{T}.
\]
\end{proof}

We also note that for every $z \in \mathbb{T}$,
\[ \left\{ \left. \dfrac{1}{\sqrt{d}} 
\left( 1, \zeta^{-1}, \zeta^{-2}, \cdots, \zeta^{1 - d} \right)^\mathrm{T} 
\ \right| \ \zeta^d = z \right\} \]
forms an orthonormal basis of $\mathbb{C}^d$.

\begin{lemma}\label{lemma: characteristic polynomial of model}
The characteristic polynomial of the inverse Fourier transform $\widehat{U_{d, \lambda}}(z)$ of the model quantum walk is
\[ \prod_{\zeta \colon \zeta^d = z} (\lambda - \lambda(\zeta)) \in \mathcal{Q}_\mathbb{T}[\lambda]. \]
\end{lemma}

\begin{proof}
The roots of the characteristic polynomial of $\widehat{U_{d, \lambda}}(z)$ 
are eigenvalues of the unitary matrix.
By Lemma \ref{lemma: eigenvectors for model},
the eigenvalues are $\{ \lambda(\zeta) \ |\ \zeta^d = z\}$.
\end{proof}


\subsection{Structure theorem}

\begin{proposition}\label{proposition: decomposition of inverse Fourier transform}
Let $\widehat{U}$ be the inverse Fourier transform of the quantum walk $U$. 
For every indecomposable system of eigenvalue functions
\[((d(1), \lambda_1), (d(2), \lambda_2), \cdots, (d(m), \lambda_m)),\]
there exists an analytic map 
$\widehat{V} \colon \mathbb{T} \to M_n(\mathbb{C})$ 
to unitary matrices
satisfying
\[\widehat{U}(z) = 
\widehat{V}(z) \left( \widehat{U_{d(1), \lambda_1}}(z) \oplus \widehat{U_{d(2), \lambda_2}}(z) \oplus \cdots \oplus \widehat{U_{d(m), \lambda_m}}(z) \right) \widehat{V}(z)^*, z \in \mathbb{T}. \]
\end{proposition}

\begin{proof}
Let
\[((d(1), \lambda_1), (d(2), \lambda_2), \cdots, (d(m), \lambda_m))\]
be an arbitrary system of eigenvalue functions of $\widehat{U}$.
By Proposition \ref{proposition: structure of dual}, 
for every $j \in \{1, 2, \cdots, m\}$
there exist
analytic maps
${\bf v}_j \colon \mathbb{T} \to \mathbb{C}^n$
satisfying that
\begin{itemize}
\item
for every $z \in \mathbb{T}$, 
\[ \left\{ {\bf v}_j (\zeta) \ 
|\ 1 \le j \le m, 
\zeta^{d(j)} = z \right\} \]
forms an orthonormal basis of $\mathbb{C}^n$,
\item
for every $1 \le j \le m$, and for every $\zeta \in \mathbb{T}$,
\[\widehat{U}\left( \zeta^{d(j)} \right) {\bf v}_j(\zeta) = \lambda_j(\zeta) {\bf v}_j(\zeta).\]
\end{itemize}
For $z \in \mathbb{T}$, define an isometric operator 
$V_j(z) \colon \mathbb{C}^{d(j)} \to \mathbb{C}^n$ by the correspondence
\[\dfrac{1}{\sqrt{d(j)}}
\left( 1, \zeta^{-1}, \zeta^{-2}, \cdots, \zeta^{1 - d(j)} \right)^\mathrm{T} 
\mapsto
{\bf v}_j(\zeta)\]
between two orthonormal systems,
where $\zeta$ is a $d(j)$-th root of $z$.
Since the unitary matrices $V_j(z)$ give a correspondence 
between analytic sections to analytic sections,
the map
$z \mapsto V_j(z)$ is analytic.

We can easily check the equation
\begin{eqnarray*}
\widehat{U}(z) V_j(z) \cdot \dfrac{1}{\sqrt{d}}
\left(1, \zeta^{-1}, \zeta^{-2}, \cdots, \zeta^{1 - d} \right) ^\mathrm{T} 
=
\widehat{U}(z) {\bf v}_j(\zeta)
= 
\lambda_j(\zeta) {\bf v}_j(\zeta).
\end{eqnarray*}
By Lemma \ref{lemma: eigenvectors for model},
we also have
\begin{eqnarray*}
& & V_j(z) \widehat{U_{d(j), \lambda_j}}(z) \cdot \dfrac{1}{\sqrt{d}}
\left(1, \zeta^{-1}, \zeta^{-2}, \cdots, \zeta^{1 - d} \right) ^\mathrm{T}\\
&=& 
\lambda_j(\zeta) V_j(z) \cdot \dfrac{1}{\sqrt{d}}
\left(1, \zeta^{-1}, \zeta^{-2}, \cdots, \zeta^{1 - d} \right) ^\mathrm{T}\\
&=& 
\lambda_j(\zeta) {\bf v}_j(\zeta).
\end{eqnarray*}
Thus we obtain $\widehat{U}(z) V_j(z) = V_j(z) \widehat{U_{d(j), \lambda_j}}(z)$.
Let 
\[\widehat{V}(z) \colon \mathbb{C}^{d(1)} \oplus \mathbb{C}^{d(2)} \oplus \cdots \oplus \mathbb{C}^{d(m)} \to \mathbb{C}^n \]
be the direct sum of $V_j, 1 \le j \le m$.
This matrix is isometric and surjective. It satisfies
\[\widehat{U}(z)  \widehat{V}(z) = 
\widehat{V}(z) \left( \widehat{U_{d(1), \lambda_1}}(z) \oplus \widehat{U_{d(2), \lambda_2}}(z) \oplus \cdots \widehat{\oplus U_{d(m), \lambda_m}}(z) \right). \]
\end{proof}

Applying the Fourier transform to Proposition
\ref{proposition: decomposition of inverse Fourier transform}, 
we obtain the following.

\begin{theorem}
[Structure theorem on analytic homogeneous quantum walks on $\mathbb{Z}$]
\label{theorem: structure theorem}
Every $n$-state discrete-time analytic homogeneous quantum walk $U$ on $\mathbb{Z}$
is conjugate to a direct sum of model quantum walks.
More precisely,
for every indecomposable system of eigenvalue functions
\[((d(1), \lambda_1), (d(2), \lambda_2), \cdots, (d(m), \lambda_m))\]
of $\widehat{U}$,
there exists an analytic unitary operator $V$ acting on $\ell_2(\mathbb{Z}) \otimes \mathbb{C}^n$
satisfying
\[U = 
V \left( U_{d(1), \lambda_1} \oplus U_{d(2), \lambda_2} \oplus \cdots \oplus U_{d(m), \lambda_m} \right) V^*. \]
\end{theorem}


%

\subsection{ Decomposable and indecomposable quantum walks }


\begin{definition}\label{definition: decomposability}
An $n$-state discrete-time analytic homogeneous quantum walk $U$
is said to be {\rm decomposable}, 
if there exist 
\begin{itemize}
\item
natural numbers $d(1)$ and $d(2)$ whose sum is $n$,
\item
discrete-time analytic homogeneous $d(1)$-state quantum walk $U_1$,
\item
discrete-time analytic homogeneous $d(2)$-state quantum walk $U_2$,
\item
and an analytic unitary $V$ acting on $\ell_2(\mathbb{Z}) \otimes \mathbb{C}^n$
\end{itemize}
satisfying
\[U = V \left( U_1 \oplus U_2 \right) V^*.\]
Otherwise, the quantum walk $U$ is said to be {\rm indecomposable}.
\end{definition}

\begin{lemma}
For every indecomposable analytic homogeneous quantum walk $U$, 
there exists an analytic unitary analytic $V$ acting on $\ell_2(\mathbb{Z}) \otimes \mathbb{C}^d$ 
and model quantum walk $U_{d, \lambda}$ such that
\[U = V U_{d, \lambda} V^*. \]
\end{lemma}

\begin{proof}
The quantum walk $U$ admits a decomposition
\[U = 
V \left( U_{d(1), \lambda_1} \oplus U_{d(2), \lambda_2} \oplus \cdots \oplus U_{d(m), \lambda_m} \right) V^* \]
described in Theorem \ref{theorem: structure theorem}.
Since the quantum walk is indecomposable, $m = 1$.
\end{proof}

\begin{proposition}\label{proposition: indecomposable QW}
A discrete-time analytic homogeneous quantum walk $U$ is indecomposable, if and only if
the characteristic polynomial $f(\lambda; z)$ of the inverse Fourier transform $\widehat{U}(z)$
is
an irreducible polynomial in $\mathcal{Q}_\mathbb{T} [\lambda]$.
\end{proposition}

\begin{proof}
Let $f(\lambda; z)$ be the characteristic polynomial of the matrix $\widehat{U}(z)$.

Suppose that $U$ is decomposed as in Definition \ref{definition: decomposability}:
\[U = V \left( U_1 \oplus U_2 \right) V^*. \]
Consider the inverse Fourier transforms of unitary operators.
We have
\[\widehat{U}(z) = 
\widehat{V}(z) \left( \widehat{U_1}(z) \oplus \widehat{U_2}(z) \right) \widehat{V}(z)^*\]
Let $f_j(\lambda; z)$ be the characteristic polynomial of the matrix $\widehat{U_j}(z)$, 
for $j = 1, 2$.
By the above decomposition, we have
\[f(\lambda; z) = f_1(\lambda; z) f_2(\lambda; z). \]
It follows that $f(\lambda) \in \mathcal{Q}_\mathbb{T}[\lambda]$ is not irreducible.

Conversely, suppose that $f(\lambda) \in \mathcal{Q}_\mathbb{T}[\lambda]$ is not irreducible.
The decomposition of $f(\lambda)$ into irreducible polynomials
\[f(\lambda) = g_1(\lambda) \cdots g_m(\lambda)\]
corresponds to an indecomposable system of eigenvalue functions 
\[((d(1), \lambda_1), (d(2), \lambda_2), \cdots, (d(m), \lambda_m))\]
of $\widehat{U}$.
Theorem \ref{theorem: structure theorem}
gives a decomposition of $U$ into model quantum walks
\[U = V \left( U_{d(1), \lambda_1} \oplus U_{d(2), \lambda_2} \oplus \cdots \oplus U_{d(m), \lambda_m} \right) V^*.\]
\end{proof}

\begin{corollary}
A model quantum walk $U_{d, \lambda}$ is decomposable, if and only if
the characteristic polynomial satisfies the rotation symmetry in the following sense:
there exists $c \in \{1, \cdots, d - 1\}$ satisfying that
\[\lambda(\exp(2 \pi i c /d)\zeta) = \lambda(\zeta), \quad \zeta \in \mathbb{T}.\]
\end{corollary}

\begin{proof}
The characteristic polynomial of the inverse Fourier transform of $U_{d, \lambda}$
is 
\[ \prod_{\zeta \colon \zeta^d = z} (\lambda - \lambda(\zeta)) \in \mathcal{Q}_\mathbb{T}[\lambda], \]
by Lemma \ref{lemma: characteristic polynomial of model}.
By Lemma \ref{lemma: iff reducible}, this polynomial is not irreducible, if and only if $\lambda$ satisfies the rotation symmetry in the above sense.
\end{proof}


\subsection{Realization by continuous-time quantum walks}

\begin{proposition}\label{proposition: realization by CTQW}
Let $\lambda \colon \mathbb{T} \to \mathbb{T}$ be an analytic map.
If the winding number of $\lambda$ is $0$, then there exists a $1$-state continuous-time analytic homogeneous quantum walk
\[\mathbb{R} \ni t \mapsto U^{(t)}\]
satisfying that $U^{(1)} = U_{1, \lambda}$.
\end{proposition}

\begin{proof}
If the winding number of $\lambda$ is $0$, then there exists an analytic function $h \colon \mathbb{T} \to \mathbb{R}$ satisfying that
\[\exp(i h(z)) = \lambda(z), \quad z \in \mathbb{T}.\]
The $1$-parameter unitary group 
\[U^{(t)} = \mathcal{F} M[ \exp (i t h)] \mathcal{F}^{-1} \]
satisfies the conditions in the proposition.
\end{proof}

Let us denote by $w(\lambda)$ 
the winding number of $\lambda \colon \mathbb{T} \to \mathbb{T}$.
By Lemma \ref{lemma: calculation of the model quantum walks},
we may factorize the model quantum walk $U_{1, \lambda}$ as follows
\[U_{1, \lambda} = U_{1, z^{w(\lambda)}} U_{1, \zeta^{-w(\lambda)} \lambda(\zeta)} = S_{w(\lambda)} U_{1, \zeta^{-w(\lambda)} \lambda(\zeta)}.\]
Since the winding number of $\zeta \mapsto \zeta^{-w(\lambda)} \lambda(\zeta)$ is $0$,
by Proposition \ref{proposition: realization by CTQW},
the model quantum walk $U_{1, \zeta^{-w(\lambda)} \lambda(\zeta)}$ can be realized by a continuous-time quantum walk.

\begin{theorem}\label{theorem: construction from CTQW}
For every $n$-state discrete-time analytic homogeneous quantum walk $U$, we can express $U$ as
\[U = V \left( \oplus_{j = 1}^m W_{d(j)}^* S_{w(j)} U^{(1)}_j W_{d(j)} \right) V^*\]
by
\begin{itemize}
\item
Preparing several $1$-state continuous-time analytic homogeneous quantum walks $\mathbb{R} \ni t \mapsto U^{(t)}_j$, $1 \le j \le m$,
\item
Restricting the group of time $\mathbb{R}$ to $\mathbb{Z}$,
\item 
Composing with shift operators $S_{w(j)} U^{(1)}_j$,
\item
Rearranging the labels of position $W_{d(j)}^* S_{w(j)} U^{(1)}_j W_{d(j)}$,
\item
Taking direct sum
\[\oplus_{j = 1}^m W_{d(j)}^* S_{w(j)} U^{(1)}_j W_{d(j)},\]
\item
Conjugacy by an analytic unitary $V$ acting on $\ell_2(\mathbb{Z}) \otimes \mathbb{C}^n$,
\end{itemize}
using some natural numbers $m$, $w(1)$, $\cdots$, $w(m)$, and $d(1)$, $\cdots$, $d(m)$.
\end{theorem}

\begin{proof}
By Theorem \ref{theorem: structure theorem},
we can describe $U$ by
\[U = 
V \left( U_{d(1), \lambda_1} \oplus U_{d(2), \lambda_2} \oplus \cdots \oplus U_{d(m), \lambda_m} \right) V^*. \]
By rearranging the labels of positions, we have
\[W_{d(j)} U_{d(j), \lambda_j} W_{d(j)}^* = U_{1, \lambda_j}. \]
By the proceeding remark before the theorem,
$U_{1, \lambda_j}$ is a product of a shift operator and a restriction of a continuous-time analytic homogeneous quantum walk.
\end{proof}

\begin{theorem}\label{theorem: realization by CTQW}
Let $U$ be an $n$-state discrete-time analytic homogeneous quantum walk. 
Let 
\[\left( (d(1), \lambda_1), (d(2), \lambda_2), \cdots, (d(m), \lambda_m) \right)\]
be an arbitrary system of eigenvalue functions of $\widehat{U}$, which is introduced in Definition \ref{definition: system of eigenvalue functions}.
The quantum walk $U$ is a restriction of a continuous-time analytic homogeneous quantum walk, 
if and only if
all the winding numbers of $\lambda_j \colon \mathbb{T} \to \mathbb{T}$ are $0$.
\end{theorem}

\begin{proof}
Suppose that $U$ can be realized by a continuous-time analytic homogeneous quantum walk.
Then for every natural number $N$, there exists an analytic homogeneous quantum walk $W$ satisfying that
$W^N = U$.
Let 
\[\left( (d(1), \rho_1), (d(2), \rho_2), \cdots, (d(M), \rho_M) \right)\]
be a system of eigenvalue functions of $\widehat{W}$ (see Definition \ref{definition: system of eigenvalue functions}).
Then
\[\left( \left(d(1), \rho_1^N\right), \left(d(2), \rho_2^N\right), \cdots, \left(d(M), \rho_M^N\right) \right)\]
is a system of eigenvalue functions of $\widehat{U}$.
(This is not necessarily an indecomposable system).
The quantity $|w|(U)$ given in Definition \ref{definition: absolute value of winding number}
satisfies
\[|w|(U) 
= \sum_{j = 1}^M \left| w \left( \rho_j^N \right) \right| 
= N \sum_{j = 1}^M |w(\rho_j)|.\]
This is an element of $N \mathbb{Z}$.
Since $N$ is arbitrary, we have $|w|(U) = 0$.
By well-definedness of $|w|(U)$, we have
\[\sum_{j = 1}^m|w(\lambda_j)| = 0.\]

Conversely, suppose that the equation
\[U = 
V \left( U_{d(1), \lambda_1} \oplus U_{d(2), \lambda_2} \oplus \cdots 
\oplus U_{d(m), \lambda_m} \right) V^*\]
holds
and that all the winding numbers $w(\lambda_j)$ is $0$.
By Proposition \ref{proposition: realization by CTQW}, 
for every $j$,
the quantum walk
$U_{d(j), \lambda_j} = W_d^* U_{1, \lambda_j} W_d$
can be realized by a continuous-time analytic homogeneous quantum walk.
\end{proof}

\section{Convergence theorem}
\label{section: convergence theorem}

\subsection{Some remarks on locality of initial unit vectors}

In the conjugacy
\[U = 
V \left( U_{d(1), \lambda_1} \oplus U_{d(2), \lambda_2} \oplus \cdots \oplus U_{d(m), \lambda_m} \right) V^* \]
given by Theorem \ref{theorem: structure theorem},
the unitary $V^*$ does not preserve finiteness of the support of the initial unit vector.
However, $V^*$ preserves weaker forms of locality.

For a unit vector $\xi \in \ell_2(\mathbb{Z}) \otimes \mathbb{C}^n$,
we define a probability measure $P[\xi]$ on $\mathbb{Z}$ by
\[ P[\xi](\{s\}) 
= \sum_{k = 1}^n| \langle \delta_s \otimes \delta_k, \xi \rangle |^2, 
\quad s \in \mathbb{Z}. \]
If there exists a positive real number $r$ greater than $1$ satisfying that
two sequences 
\[ \left( r^{s} P[\xi](\{s\}) \right)_{s \in \mathbb{Z}}, \quad
 \left( r^{-s} P[\xi](\{s\}) \right)_{s \in \mathbb{Z}} \]
are elements in $\ell_1(\mathbb{Z})$,
then we call the probability measure $P[\xi](\{s\})$ and the unit vector $\xi$
are {\it of the exponential type}.
In the study of quantum walks, this type of probability measures have attracted attention. See \cite{EndoKonnoYMJ}, for example.
This is equivalent to the condition that the inverse Fourier transform $(\mathcal{F}^{-1} \otimes \mathrm{id}) \xi \in L^2(\mathbb{T}, \mathbb{C}^n)$ is analytic on $\mathbb{T}$. 

To apply differential operators on the Fourier dual,
we consider wider class.
In the case that
for every natural number $d$,
the sequence
\[ \left\{ (1 + s^2)^d P[\xi](\{s\}) \right\}_{s = - \infty}^\infty\]
is bounded,
we call the probability measure $P[\xi](\{s\})$ and the unit vector $\xi$
{\it rapidly decrease}.
This is equivalent to the condition that $(\mathcal{F}^{-1} \otimes \mathrm{id}) \xi$ is  smooth on $\mathbb{T}$. 


If $\xi \in \ell_2(\mathbb{Z}) \otimes \mathbb{C}^n$ is of the exponential type and if the operator $V$ acting on $\ell_2(\mathbb{Z}) \otimes \mathbb{C}^n$ is analytic, then $V^* \xi$ is of the exponential type.
If $\xi \in \ell_2(\mathbb{Z}) \otimes \mathbb{C}^n$ rapidly decreases and if the operator $V$ acting on $\ell_2(\mathbb{Z}) \otimes \mathbb{C}^n$ is analytic, then $V^* \xi$ rapidly decreases.

%
%
%

\subsection{Limit distributions}

Let $U$ be a discrete-time analytic homogeneous quantum walk acting on $\ell_2(\mathbb{Z}) \otimes \mathbb{C}^n$.
Given a unit vector $\xi$ called an initial vector, we obtain a sequence of probability measures on $\mathbb{Z}$
\[P[\xi], P[U \xi], P[U^2 \xi], P[U^3 \xi], \cdots.\]
For every time $t \in \mathbb{Z}$, consider the pushforward
\[ \phi^{(t)}_* (P[U^t \xi]) \in {\rm Prob}(\mathbb{R})\]
with respect to the map $\mathbb{Z} \ni s \mapsto s /t \in \mathbb{R}$.
The goal of this section is to show that the sequence $\{ \phi^{(t)}_* (P[U^t \xi]) \}_{t = 1}^\infty$ weakly converges.

To study the limit distribution, we use the following diagonal self-adjoint operator on $\ell_2 (\mathbb{Z})$:
\[\dfrac{D}{t} \colon \delta_s \mapsto \dfrac{s}{t} \delta_s, \quad s \in \mathbb{Z}.\]
Its inverse Fourier transform is a self-adjoint operator on $L^2 (\mathbb{T})$
\[\mathcal{F}^{-1} \dfrac{D}{t} \mathcal{F} \colon z^s \mapsto s z^s, \quad s \in \mathbb{Z}.\]
When we write $\lambda$ as $\lambda(z)$,
the operator is identical to the differential operator $\dfrac{1}{t} z \dfrac{d}{dz}$.
When we write $\lambda$ as $\lambda(e^{i \theta})$,
the operator is identical with the differential operator $\dfrac{1}{it} \dfrac{d}{d \theta}$.
In the case that the probability measure $P[\xi]$ on $\mathbb{Z}$ rapidly decreases,
the vector $\xi$ is in the domain of $\left( \dfrac{D}{t} \right)^m$.

\begin{lemma}\label{lemma: diagonal operator}
Let $\xi$ be 
a unit vector in $\ell_2 (\mathbb{Z}) \otimes \mathbb{C}^n$ which rapidly decreases.
Then for every natural number $m$,
the $m$-th moment of $\phi^{(t)}_* (P[\xi])$ is finite,
and the moment is given by the following formula:
\[ \left\langle \left( \dfrac{D}{t} \otimes \mathrm{id} \right)^m \xi, \xi \right\rangle. \]
\end{lemma}

\begin{proof}
We express the vector $\xi$ by
\[\sum_{s \in \mathbb{Z}, k \in \{1, \cdots, n\}} \xi(s, k) \delta_s \otimes \delta_k. \]
We compute the inner product as follows:
\begin{eqnarray*}
\left\langle \left( \dfrac{D}{t} \otimes \mathrm{id} \right)^m \xi, \xi \right\rangle
&=&
\sum_{s \in \mathbb{Z}, k \in \{1, \cdots, n\}} \left( \dfrac{s}{t} \right)^m |\xi(s, k)|^2\\
&=&
\sum_{s \in \mathbb{Z}} \left( \dfrac{s}{t} \right)^m \sum_{k \in \{1, \cdots, n\}} |\xi(s, k)|^2\\
&=&
\sum_{s \in \mathbb{Z}} \left( \dfrac{s}{t} \right)^m P[\xi](\{s\}).
\end{eqnarray*}
The last quantity is the $m$-th moment of $\phi^{(t)}_* (P[\xi])$. 
\end{proof}

\begin{lemma}\label{lemma: convergence theorem for U_1}
Let $\lambda \colon \mathbb{T} \to \mathbb{T}$ be an analytic map.
Let $U$ be the model quantum walk $U_{1, \lambda}$. 
Let $\xi \in \ell_2 (\mathbb{Z})$ be a unit vector which rapidly decays.
Then the sequence $\left\{ \phi^{(t)}_* (P[U^t \xi]) \right\}_{t = 1}^\infty$ weakly converges to a measure whose support is compact.
\end{lemma}

\begin{proof}
We prove that the $m$-th moments of measures $\{ \phi^{(t)}_* (P[U^t \xi]) \}_{t = 1}^\infty$ converge, for every $m$.
By Lemma \ref{lemma: diagonal operator}, the moment is given by
\begin{eqnarray*}
\left\langle \left( \dfrac{D}{t} \right)^m U^t \xi, U^t \xi \right\rangle
&=&
\left\langle \left( U^{-t} \dfrac{D}{t} U^t \right)^m \xi, \xi \right\rangle.
\end{eqnarray*}
Recall that $U^t = U_{1, \lambda}^t$ is the Fourier transform $\mathcal{F} M[\lambda^t] \mathcal{F}^{-1}$ of the multiplication operator by $\lambda^t \in C(\mathbb{T})$.
The operator $U^{-t} \dfrac{D}{t} U^t$ is equal to the following
\[U^{-t} \dfrac{D}{t} U^t 
= \mathcal{F} M[\lambda^{-t}] \dfrac{1}{i t} \dfrac{d}{d \theta}M[\lambda^t] \mathcal{F}^{-1}
= \mathcal{F} \left( 
\dfrac{1}{i t} \dfrac{d}{d \theta} + M \left[ \dfrac{1}{i \lambda} \dfrac{d \lambda} {d\theta} \right] 
\right) \mathcal{F}^{-1}.
\]
Note that the function $\dfrac{1}{i \lambda (e^{i \theta})} \dfrac{d \lambda (e^{i \theta})} {d\theta}$ is the derivative of $\arg(\lambda(e^{i \theta}))$. 
We denote the function by $h(\theta)$.
This is an analytic real-valued function on $\mathbb{T}$.
The $m$-th moment of $\phi^{(t)}_* (P[U^t \xi])$ is given by
\begin{eqnarray*}
\left\langle \left( \dfrac{1}{i t} \dfrac{d}{d \theta} + M \left[ h \right] 
 \right)^m \mathcal{F}^{-1} \xi, 
\mathcal{F}^{-1} \xi \right\rangle_{L^2(\mathbb{T})}.
\end{eqnarray*}
Note that $\mathcal{F}^{-1} \xi$ is smooth, 
because its Fourier coefficients rapidly decreases.
As $t$ tends to $\infty$, the sequence of the $m$-th moment converges to
\begin{eqnarray*}
\left\langle 
\left( M \left[ h \right] \right)^m 
\mathcal{F}^{-1} \xi, 
\mathcal{F}^{-1} \xi \right\rangle_{L^2(\mathbb{T})}
= \int_\mathbb{T} h(e^{i \theta})^m 
\left| [\mathcal{F}^{-1} \xi] (e^{i \theta}) \right|^2 \dfrac{d \theta}{2 \pi}.
\end{eqnarray*}

Let $\mu = \mu_{U, \xi}$ be the pushforward of the probability measure 
$\left| [\mathcal{F}^{-1} \xi] (e^{i \theta}) \right|^2 \dfrac{d \theta}{2 \pi}$ on $\mathbb{T}$ 
by the analytic map $h \colon \mathbb{T} \to \mathbb{R}$. 
The above integral can be written as follows
\[\int_{\mathbb{R}} r^m d \mu (r).\]
Thus we have the moment convergence of the sequence $\{ \phi^{(t)}_* (P[U^t \xi]) \}_{t = 1}^\infty$ to the measure $\mu$.

Since the map $h \colon \mathbb{T} \to \mathbb{R}$ is continuous,
the support of $\mu$ is compact.
The moment convergence to $\mu$ means weak convergence.
\end{proof}

\begin{lemma}\label{lemma: convergence theorem for U_d}
Let $d$ be a natural number and 
let $\lambda \colon \mathbb{T} \to \mathbb{T}$ be an analytic map.
Let $U$ be the model quantum walk $U_{d, \lambda}$. 
Let $\xi \in \ell_2 (\mathbb{Z}) \otimes \mathbb{C}^d$ be a rapidly decreasing unit vector.
Then
the sequence $\left\{ \phi^{(t)}_* (P[U^t \xi]) \right\}_{t = 1}^\infty$ weakly converges to a measure whose support is compact.
\end{lemma}

\begin{proof}
Recall that  the model quantum walk $U_{d, \lambda}$ is conjugate with the $1$-state model quantum walk $U_{1, \lambda}$ by
\[U_{d, \lambda} = W_d^* U_{1, \lambda} W_d.\]
By Lemma \ref{lemma: convergence theorem for U_1},
the sequence of measures $\{ \phi^{(t)}_* (P[U_{1, \lambda}^t W_d \xi]) \}_{t = 1}^\infty$ weakly converge.
Denote by $\eta^{(t)}$ the unit vector $U_{1, \lambda}^t W_d \xi$.
Compare the probability measure $\phi^{(t)}_* P[W_d^* \eta^{(t)}]$ with $\phi^{(t)}_* P[\eta^{(t)}]$.
The former measure is the pushforward of the latter measure under the map
\[\mathbb{Z}[1/t] \ni (ds + k)/t \mapsto s/t \in \mathbb{Z}[1/t], \quad k \in \{1, \cdots, d\}.\] 
If $t$ is large, the above map is approximated by
\[\mathbb{R} \ni r \mapsto r/d \in \mathbb{R}.\]
Therefore, the sequence $\phi^{(t)}_* P[W_d^* \eta^{(t)}]$ also weakly converges, and the limit is the pushforward of the limit of  $\phi^{(t)}_* P[\eta^{(t)}]$ by the map
$r \mapsto r/d$.
\end{proof}

\begin{theorem}\label{theorem: limit distribution}
For every discrete-time analytic homogeneous quantum walk $U$, and for every rapidly decreasing initial unit vector $\xi$,
the sequence of probability measures $\left\{ \phi^{(t)}_* (P[U^t \xi]) \right\}_{t = 1}^\infty$ on $\mathbb{R}$ weakly converges to a measure whose support is compact.
\end{theorem}

\begin{proof}
By Theorem \ref{theorem: structure theorem}, the quantum walk $U$ is presented by an analytic unitary operator $V$ and model quantum walks as follows:
\[
U = 
V \left( U_{d(1), \lambda_1} \oplus U_{d(2), \lambda_2} \oplus \cdots \oplus U_{d(m), \lambda_m} \right) V^*. \]
Note that the vector $V^* \xi$ rapidly decreases.
Denote by $W^{t}$ the unitary
\[
\left( U_{d(1), \lambda_1} \oplus U_{d(2), \lambda_2} \oplus \cdots \oplus U_{d(m), \lambda_m} \right)^t.
\]
By Lemma \ref{lemma: convergence theorem for U_d}, the sequence of measures 
\[
\phi^{(t)}_* P[W^t V^* \xi ]
\]
weakly converges.

The $m$-th moment of the probability measure $\phi^{(t)}_* (P[U^t \xi])$ is written by
\begin{eqnarray*}
\left\langle \left( \dfrac{D}{t} \otimes \mathrm{id} \right)^m U^t \xi, U^t \xi \right\rangle
&=&
\left\langle \left( W^{-t} V^* \left( \dfrac{D}{t} \otimes \mathrm{id}\right) V W^t \right)^m V^* \xi, V^* \xi \right\rangle.
\end{eqnarray*}
We consider the commutator
\[\left( \dfrac{D}{t} \otimes \mathrm{id}\right) V - V \left( \dfrac{D}{t} \otimes \mathrm{id}\right). \]
Its inverse Fourier transform is
\[\left( \dfrac{z}{t} \dfrac{d}{dz} \otimes \mathrm{id}\right) \widehat{V} 
- \widehat{V} \left( \dfrac{z}{t} \dfrac{d}{dz} \otimes \mathrm{id}\right). \]
Note that every entry $\widehat{V}_{k, l}(z)$ of $\widehat{V}$ gives a multiplication operator by an analytic function on $\mathbb{T}$.
By the Leibniz rule of differential, the above operator is simply given by the multiplication operator by
\[\dfrac{1}{t} \left( z \dfrac{d \widehat{V}_{k, l}(z)}{dz} \right)_{k, l}.\]
As $t$ becomes large, the operator norm converges to $0$.
It follows that the $m$-th moment
\begin{eqnarray*}
\left\langle \left( W^{-t} V^* \left( \dfrac{D}{t} \otimes \mathrm{id}\right) V W^t \right)^m V^* \xi, V^* \xi \right\rangle.
\end{eqnarray*}
is asymptotically equal to  
\begin{eqnarray*}
\left\langle \left( W^{-t} \left( \dfrac{D}{t} \otimes \mathrm{id}\right) W^t \right)^m V^* \xi, V^* \xi \right\rangle
=
\left\langle \left( \dfrac{D}{t} \otimes \mathrm{id} \right)^m W^t V^* \xi, W^t V^* \xi \right\rangle
\end{eqnarray*}
This is the $m$-th moment of $\phi^{(t)}_* P[W^t V^* \xi ]$.

Therefore the sequence of measures $\{ \phi^{(t)}_* (P[U^t \xi]) \}_{t = 1}^\infty$ converges in moments and the limit distribution has compact support.
\end{proof}

\begin{remark}
\label{remark: our advantage}
In the following point, our argument improves the known study of space-homogeneous quantum walks on $\mathbb{Z}$.
\begin{itemize}
\item
We clarified the definition of quantum walks on $\mathbb{Z}$ 
in Section \ref{section: analyticity}.
\item
As proceeding studies have said, the eigenvalue functions $\lambda_1, \lambda_2$ are smooth.
However, this is not a trivial claim. The claim is proved in Proposition \ref{proposition: eigenvalue functions}.
In fact, the eigenvalue functions are analytic.
\item
Our argument covers the case that the eigenvalue function of the inverse Fourier transform is not single-valued.
Subsection \ref{subsection: modified Hadamard walk} gives an example.
In the authors opinion, it sounds natural to say that the unitary operator $U$ in Subsection \ref{subsection: modified Hadamard walk} is a quantum walk, because in the actual experiment the walk is substantially identical to a (part of) usual quantum walk.
The paper \cite{GJS} implicitly concentrate on the case that the eigenvalue functions are single-valued.
\item
Our argument encompasses the case that the characteristic polynomial has multiple roots.
The paper \cite{GJS} does not successfully explain why we only have to consider the case that there exists no multiple root.
\end{itemize}

\end{remark}

\section{Classification of analytic unitary operators in ${\rm C}^\ast_{\rm red}(\mathbb{Z}) \otimes M_n(\mathbb{C})$ }

Let us recall that the reduced group C$^*$-algebra ${\rm C}^*_{\rm red}(\mathbb{Z})$
of $\mathbb{Z}$ is the operator norm closure of 
$\sum_{s \in \mathbb{Z}} \mathbb{C} S_s \subset \mathbb{B}(\ell_2 (\mathbb{Z}))$.
For a complex number $z$ with modulus $1$, the linear map defined by
\[S_s \mapsto z^s S_s, \quad x \in \mathbb{Z}\]
extends to an automorphism $\alpha_z$ on C$^\ast_{\rm red}(\mathbb Z)$.
For every operator $X$ in C$^\ast_{\rm red}(\mathbb Z)$, the map
\[X \mapsto \alpha_z(X)\]
is continuous with respect to operator norm topology.
The action 
\[\alpha \colon \mathbb{T} \to {\rm Aut} ({\rm C}^\ast_{\rm red}(\mathbb Z))\]
is called the gauge action.
Note that $\alpha$ naturally extends to ${\rm C}^\ast_{\rm red}(\mathbb Z) \otimes M_n(\mathbb{C})$.

For a homogeneous operator $X$ on $\ell_2(\mathbb{Z}) \otimes \mathbb{C}^n$,
the operator $X$ is analytic,
if and only if $X$ is an element of ${\rm C}^\ast_{\rm red}(\mathbb Z) \otimes M_n(\mathbb{C})$ and 
there exist 
\begin{enumerate}
\item
a domain $\Omega \subseteq \mathbb{C}$ containing the unit circle $\mathbb{T}$,
\item
and an analytic map 
$\Omega \to {\rm C}^\ast_{\rm red}(\mathbb Z) \otimes M_n(\mathbb{C})$
which extends the map 
$\mathbb{T} \ni z \mapsto \alpha_z(X)$
given by the gauge action.
\end{enumerate}

Let $\mathcal{U}(n)$ be the set of all the analytic unitary operators in ${\rm C}^\ast_{\rm red}(\mathbb{Z}) \otimes M_n(\mathbb{C})$. 
We introduce an equivalence relation $\sim$ on $\mathcal{U}(n)$ by conjugacy.
More precisely, 
two elements $U_1$ and $U_2$ are said to be equivalent, 
if there exists $V \in \mathcal{U}(n)$ such that
$V U_1 V^* = U_2$.
Our argument provide a classification result on $\mathcal{U}(n)$ up to conjugacy.

Let $\mathcal{E}(n)$ be the collections of all the indecomposable system of eigenvalue functions
\[((d(1), \lambda_1), (d(2), \lambda_2), \cdots, (d(m), \lambda_m))\]
introduced in Subsection \ref{subsection: eigenvalue functions}.
Here, 
\begin{itemize}
\item
$m$ is a natural number, 
\item
$d(1)$, $\cdots$, $d(m)$ are natural numbers whose sum is $n$, 
\item
and $\lambda_1$, $\cdots$, $\lambda_m$ 
are analytic maps from $\mathbb{T}$ to $\mathbb{T}$ 
such that for any $j$, $\lambda_j$ does not satisfy the conditions on reducibility 
explained in Lemma \ref{lemma: iff reducible}.
\end{itemize}
We introduce an equivalence relation $\sim$ on $\mathcal{E}(n)$ by the procedures
(1), (2) explained in Subsection \ref{subsection: eigenvalue functions}.

\begin{theorem}
The map 
\begin{eqnarray*}
& & ((d(1), \lambda_1), (d(2), \lambda_2), \cdots, (d(m), \lambda_m))\\
&\mapsto&
U_{d(1), \lambda_1} \oplus U_{d(2), \lambda_2} \oplus \cdots \oplus U_{d(m), \lambda_m}
\end{eqnarray*}
from $\mathcal{U}(n)$ to $\mathcal{E}(n)$ 
induces a bijective correspondence between
$\mathcal{E}(n) / \sim$ and $\mathcal{U}(n) / \sim$. 
\end{theorem}

\begin{proof}
For a natural number $d$ and an analytic map 
$\lambda \colon \mathbb{T} \to \mathbb{T}$,
consider the rotation $\widetilde{\lambda}$ by
\[\widetilde{\lambda}(\zeta) = \lambda(\exp(2 \pi i / d) \zeta).\]
We can easily show that
$U_{d, \widetilde{\lambda}}$ is conjugate to $U_{d, \lambda}$,
by the definition of model quantum walks introduced 
in Subsection \ref{subsection: model quantum walk}.
It follows that the map from 
$\mathcal{E}(n) / \sim$ to $\mathcal{U}(n) / \sim$ is well-defined. 

The map is surjective, by Theorem \ref{theorem: structure theorem}.

We next prove that the map is injective.
For the unitary 
\[
U =
U_{d(1), \lambda_1} \oplus U_{d(2), \lambda_2} \oplus \cdots \oplus U_{d(m), \lambda_m},
\]
consider the inverse Fourier transform $\widehat{U}(z)$.
The characteristic function is given by
\[f(\lambda; z) 
= \prod_{j = 1}^m \prod_{\zeta \colon \zeta^{d(j)} = z} (\lambda - \lambda_j(\zeta)).\]
By Lemma \ref{lemma: iff reducible}, each factor $\prod_{\zeta \colon \zeta^{d(j)} = z} (\lambda - \lambda_j(\zeta))$
is an irreducible polynomial in $\mathcal{Q}_\mathbb{T}[\lambda]$.
By uniqueness of the irreducible decomposition,
the system of eigenvalue functions
\[((d(1), \lambda_1), (d(2), \lambda_2), \cdots, (d(m), \lambda_m))\]
is uniquely determined up to permutation on $\{1, \cdots, m\}$ and rotation on $\lambda_j(\zeta)$ by $d(j)$-th root of $1$.
It follows that the map from 
$\mathcal{E}(n) / \sim$ to $\mathcal{U}(n) / \sim$ is injective.
\end{proof}

\section{An algebra to which a quantum walk belongs}\label{section: algebraic definition}

\begin{theorem}
Every $n$-state discrete-time analytic homogeneous quantum walk $U$ on $\mathbb{Z}$ is a solution of some algebraic equation of degree $n$ whose coefficients are analytic elements of ${\rm C}^\ast_{\rm red}(\mathbb{Z}) \otimes \mathbb{C} \subset \mathbb{B}(\ell_2(\mathbb{Z} \otimes \mathbb{C}^n)$.
\end{theorem}

\begin{proof}
The characteristic polynomial $f(\lambda; z)$ of $\widehat{U}(z)$ satisfies
\[f(\widehat{U}(z); z) = O.\]
The left hand side is a polynomial of $\widehat{U}(z)$ whose coefficients are analytic functions on $\mathbb{T}$.
Applying the inverse Fourier transform to the equation, we obtain the theorem.
\end{proof}

\begin{remark}
Let us recall the following basic insight emphasized by Mikio Sato:
\begin{itemize}
\item
Equation $\leftrightarrow$ Algebra (or Module)
\item
Solution $\leftrightarrow$ Homomorphism
\end{itemize}
For our concrete quantum walk $U$, let $f(\lambda) \in \mathcal{Q}_\mathbb{T}[\lambda]$ be the characteristic polynomial of the inverse Fourier transform $\widehat{U}(z)$.
We have the following:
\begin{itemize}
\item
Algebraic equation $f= 0$ defines a quotient algebra $\mathcal{Q}_\mathbb{T}[\lambda]/(f)$,
where $(f)$ stands for the ideal generated by $f$.
\item
The solution $U$ of the algebraic equation $f = 0$ defines a homomorphism defined on $\mathcal{Q}_\mathbb{T}[\lambda] / (f)$ such that $\lambda + (f)$ maps to $U$.
\end{itemize}
There seems to be alternative framework of quantum walks in which we can treat quantum walks more algebraically. 
Once such kind of framework is established, a concrete form $U$ of quantum walk will be regarded as an image of a homomorphism from some algebra.
\end{remark}

\section{Examples}
\label{section: examples}

As in the previous sections, for an integer $s \in \mathbb{Z}$, $S_s$ stands for the shift operator $\delta_t \to \delta_{s + t}$ on $\ell_2(\mathbb{Z})$.

\subsection{Some $2$-state quantum walk}

Let $a$, $b$ be complex numbers satisfying $|a|^2 + |b|^2 = 1, ab \neq 0$.
We express $a$ and $b$ as follows
\[a = r e^{i \alpha}\,\quad b = \sqrt{1 - r^2} e^{i \beta},\]
where $\alpha$, $\beta$, and $0 < r < 1$ are real numbers.
Let us consider the unitary operator
\[
U = 
\left(
\begin{array}{cc}
\overline{a} S_{-1} & - b S_{-1} \\
\overline{b} S_1 & a S_1 
\end{array}
\right), \quad z \in \mathbb{T}
\]
acting on $\ell_2(\mathbb{Z}) \otimes \mathbb{C}^2$.
We regard $\ell_2(\mathbb{Z}) \otimes \mathbb{C}^2$ as the set of column vectors of length $2$ whose entries are $\ell_2$ functions on $\mathbb{Z}$.
The weak limit theorem for this walk has been already shown in \cite{KonnoJMSJ}.

Let us determine whether this walk is a restriction of a continuous-time quantum walk
and whether it is indecomposable.
For simplicity, we assume that $\alpha = 0$.
The characteristic function of the inverse Fourier transform $\widehat{U}(z)$ is
\[f(\lambda; z) = \lambda^2 - r \left( z + z^{-1} \right) \lambda + 1.\]
We express $z$ by $e^{i \theta}$.
The roots are
\begin{eqnarray*}
\lambda_1(e^{i \theta}) &=& r \cos \theta + i \sqrt{1 - r^2 \cos^2 \theta},\\
\lambda_2(e^{i \theta}) &=& r \cos \theta - i \sqrt{1 - r^2 \cos^2 \theta}.
\end{eqnarray*}
They are single-valued functions.
It follows that
\[((1, \lambda_1), (1, \lambda_2))\]
is an indecomposable system of eigenvalue functions of $\widehat{U}$ introduced 
in Subsection \ref{subsection: eigenvalue functions}.
The characteristic function is expressed as follows:
\[f(\lambda; z) = (\lambda - \lambda_1(z)) (\lambda - \lambda_2(z)).\]
By Proposition \ref{proposition: indecomposable QW}, the walk is decomposable.
The winding numbers of the eigenvalue functions are $0$.
By Theorem \ref{theorem: realization by CTQW}, 
the walk is a restriction of a continuous-time analytic homogeneous quantum walk.

We now proceed to calculate more concrete description of the continuous-time quantum walk.
For the inverse Fourier transform of $U$, we have
\begin{eqnarray*}
 & & 
\widehat{U}(e^{i \theta}) - \lambda_2(e^{i \theta})
\\
 &=& 
\left(
\begin{array}{cc}
i \sqrt{1 - r^2 \cos^2 \theta} - i r \sin \theta & - \sqrt{1 - r^2} e^{- i(\theta - \beta)} \\
\sqrt{1 - r^2} e^{i(\theta - \beta)} & i \sqrt{1 - r^2 \cos^2 \theta} + i r \sin \theta
\end{array}
\right)
\\
 & & 
\widehat{U}(e^{i \theta}) - \lambda_1(e^{i \theta})
\\
 &=& 
\left(
\begin{array}{cc}
- i \sqrt{1 - r^2 \cos^2 \theta} - i r \sin \theta & - \sqrt{1 - r^2} e^{- i(\theta - \beta)} \\
\sqrt{1 - r^2} e^{i(\theta - \beta)} & - i \sqrt{1 - r^2 \cos^2 \theta} + i r \sin \theta
\end{array}
\right)
\end{eqnarray*}
The first column vectors of these matrices 
\begin{eqnarray*}
{\bf x}_1(e^{i \theta}) &=& 
\left(
\begin{array}{c}
i \sqrt{1 - r^2 \cos^2 \theta} - i r \sin \theta \\
\sqrt{1 - r^2} e^{i(\theta - \beta)}
\end{array}
\right)
\\
{\bf x}_2(e^{i \theta}) &=& 
\left(
\begin{array}{cc}
- i \sqrt{1 - r^2 \cos^2 \theta} - i r \sin \theta\\
\sqrt{1 - r^2} e^{i(\theta - \beta)}
\end{array}
\right)
\end{eqnarray*}
are eigenvectors of eigenvalues $\lambda_1(e^{i \theta})$, $\lambda_2(e^{i \theta})$, respectively. 
These normalizations
${\bf v}_1(e^{i \theta})$, ${\bf v}_2(e^{i \theta})$ form an orthonormal basis.

Define a unitary matrix $\widehat{V}(e^{i \theta})$ by
\[\widehat{V}(e^{i \theta}) = \left( {\bf v}_1(e^{i \theta})\quad {\bf v}_2(e^{i \theta}) \right).\]
The inverse Fourier transform of the walk is diagonalized as follows:
\begin{eqnarray*}
\widehat{U}(e^{i \theta})
=
\widehat{V}(e^{i \theta})
\left(
\begin{array}{cc}
\lambda_1(e^{i \theta}) & 0 \\
0 & \lambda_2(e^{i \theta})
\end{array}
\right)
\widehat{V}(e^{i \theta})^*.
\end{eqnarray*}
There exists a real-valued function 
$h \colon \mathbb{T} \to \mathbb{R}$ satisfying
\[
\exp(i h(e^{i \theta})) = \lambda_1(e^{i \theta}), \quad
\exp(- i h(e^{i \theta})) = \overline{\lambda_1(e^{i \theta})} 
= \lambda_2(e^{i \theta}).
\]
The walk can be obtained by restricting a continuous-time analytic homogeneous quantum walk $t \mapsto U^{(t)}$ whose inverse Fourier transform is
\begin{eqnarray*}
\widehat{U^{(t)}}(e^{i \theta})
=
\widehat{V}(e^{i \theta})
\left(
\begin{array}{cc}
\exp(i t h(e^{i \theta})) & 0 \\
0 & \exp(- i t h(e^{i \theta}))
\end{array}
\right)
\widehat{V}(e^{i \theta})^*.
\end{eqnarray*}

%
%
%

\subsection{Another type of $2$-state quantum walk}
\label{subsection: modified Hadamard walk}
In this subsection, we consider a quantum walk defined by
\[U
=
\left(
\begin{array}{cc}
r S_1 & - b S_1 \\
\overline{b} & r
\end{array}
\right), \quad r \in \mathbb{R}, b \in \mathbb{C}, r^2 + |b|^2 = 1.\]
acting on $\ell_2(\mathbb{Z}) \otimes \mathbb{C}^2$.
The weak limit theorem for this walk can be easily deduced from that of the last subsection, but to observe the eigenvalue function, we proceed.

The characteristic function of the inverse Fourier transform $\widehat{U}(z)$ is
\[f(\lambda; z) = \lambda^2 - r \left( z + 1 \right) \lambda + z.\]
Define an analytic function $\lambda_1(\zeta) \colon \mathbb{T} \to \mathbb{T}$ by
\[\lambda_1(e^{i \theta}) 
= r e^{i \theta} \cos \theta + e^{i \theta} \sqrt{1 - r^2 \cos^2 \theta}\]
The characteristic function is expressed as follows:
\[f(\lambda; z) = \prod_{\zeta \colon \zeta^2 = z} (\lambda - \lambda_1(\zeta)).\]
The calculation is straightforward.
Since $\lambda_1(- \zeta) \neq \lambda_1(\zeta)$, the polynomial $f(\lambda) \in \mathcal{Q}_\mathbb{T}[\lambda]$ is irreducible, by Lemma \ref{lemma: iff reducible}.
For the quantum walk $U$,
$((2, \lambda_1))$ is an indecomposable system of eigenvalue functions defined in Subsection \ref{subsection: eigenvalue functions}.
By Proposition \ref{proposition: indecomposable QW}, the quantum walk is indecomposable.
The winding number of the eigenvalue function $\lambda_1$ is $1$.
By Theorem \ref{theorem: realization by CTQW}, 
the quantum walk is not a restriction of a continuous-time analytic homogeneous quantum walk.

\subsection{$3$-state Grover walk}
\label{subsection: $3$-state Grover walk}

The $3$-state Grover walk is given by 
\[
U = 
\frac{1}{3}
\left(
\begin{array}{ccc}
- S_{-1} & 2 S_{-1} & 2 S_{-1}\\
2 & -1 & 2\\
2 S_1 & 2 S_1 & -S_1
\end{array}
\right).
\]
The weak limit theorem for this walk was given in \cite{InuiKonnoSegawa}.
The characteristic function of the inverse Fourier transform $\widehat{U}(z)$ is
\begin{eqnarray*}
f(\lambda; z) 
&=& \lambda^3 + \dfrac{1}{3} (z + 1 + z^{-1}) \lambda^2 
- \dfrac{1}{3} (z + 1 + z^{-1}) \lambda - 1\\
f(\lambda; e^{i \theta})
&=& 
\left( \lambda^2 + 2 \dfrac{2 + \cos \theta}{3} \lambda + 1 \right) (\lambda - 1).
\end{eqnarray*}
The eigenvalues of $\widehat{U} \left( e^{i \theta} \right)$ are
\[\quad \lambda_1 = - \dfrac{2 + \cos \theta}{3} 
\pm i \sqrt{1 - \left( \dfrac{2 + \cos \theta}{3} \right)^2}, \quad \lambda_2 = 1.\]
Let us observe the first eigenvalues $\lambda_1$.
The eigenvalues do not define single-valued analytic functions.
The eigenvalues $\lambda_1$ of $\widehat{U} \left( e^{2 i \theta} \right)$ are
\[\quad \lambda_1 = - \dfrac{2 + \cos 2 \theta}{3} 
\pm \dfrac{2i}{3} \sin \theta \sqrt{3 - \sin \theta ^2}.\]
If we add $\pi$ to $\theta$, then we get the other root.
We choose
\begin{eqnarray*}
- \dfrac{2 + \cos 2 \theta}{3} + \dfrac{2}{3} i \sin \theta \sqrt{3 - \sin^2 \theta},
\end{eqnarray*}
as the definition of the eigenvalue function $\lambda_1(e^{i \theta})$. 
The eigenvalues of $\widehat{U} \left( z \right)$ are
\[\{\lambda_1(\zeta) \ |\ \zeta^2 = z\} \cup \{1\}.\]
The pair $((2, \lambda_1), (1, 1))$ is an indecomposable system of eigenvalues.
The characteristic function is expressed as follows:
\[f(\lambda; z) = (\lambda - 1) \cdot \prod_{\zeta \colon \zeta^2 = z} (\lambda - \lambda_1(\zeta)).\]

By Proposition \ref{proposition: indecomposable QW}, 
the $3$-state Grover walk is decomposable.
The winding numbers of the eigenvalue functions are both $0$.
By Theorem \ref{theorem: realization by CTQW}, 
the walk is a restriction of a continuous-time analytic homogeneous quantum walk.

\bibliographystyle{amsalpha}
\bibliography{structureofQW.bib}

\end{document}